\documentclass[11pt]{lmcs} 
\pdfoutput=1

\usepackage[utf8]{inputenc}
\usepackage{lastpage}
\lmcsdoi{16}{3}{6}
\lmcsheading{}{\pageref{LastPage}}{}{}%
{Sep.~11,~2018}{Jul.~15,~2020}{}


\usepackage{tikz}
  \usetikzlibrary{calc}

\usepackage{stmaryrd,amssymb,proof,cmll,mathrsfs}
\usepackage{todonotes}
\theoremstyle{plain} 
\newcommand{\nat}{\mathbb{N}}
\newcommand{\V}{\mathsf{v}}

\newcommand{\R}{\mathsf{r}}
\newcommand\seq\vec
\newcommand{\st}{\mid}
\newcommand{\set}[1]{\{#1\}}
\newcommand{\bnfeq}{::=}
\newcommand{\rel}[1]{\mathsf{#1}}
\newcommand{\lsv}{\lam^\sigma_v}
\newcommand{\Var}{\mathbb{V}}
\newcommand{\Val}[1][]{\mathrm{Val}^{#1}}
\newcommand{\CbN}{CbN}
\newcommand{\CbV}{CbV}
\newcommand{\comb}[1]{\mathtt{#1}}
\newcommand{\lam}{\ensuremath{\lambda}}
\newcommand{\Lam}[1][]{\Lambda^{#1}}
\newcommand{\FV}[1]{\mathrm{FV}(#1)}
\newcommand{\subst}[2]{{[}#1 := #2]}
\newcommand*{\msto}[1][\V]{\twoheadrightarrow_{#1}}
\newcommand{\BT}[2][]{\mathrm{BT}^{#1}(#2)}
\newcommand{\nf}[2][\V]{\mathrm{nf}_{#1}(#2)}
\newcommand{\hole}[1]{\llparenthesis #1 \rrparenthesis}
\newcommand{\obseq}{\equiv}
\newcommand{\App}{\mathcal{A}}
\newcommand{\Appof}[1]{\mathcal{A}(#1)}
\newcommand{\AM}{\Appof M}
\newcommand{\BTle}{\sqsubseteq}
\newcommand{\BTge}{\sqsupseteq}
\newcommand{\Lamb}{\Lambda_\bot}
\newcommand{\Sup}{\bigsqcup}
\newcommand{\hgt}[1]{\mathsf{ht}(#1)}
\newcommand\cn[1]{\mathsf{c}_{#1}}
\newcommand{\Pair}[1]{{[}#1{]}}
\newcommand\Tuple[1]{\langle #1 \rangle}
\newcommand\Eval{\mathsf{E}\,}
\newcommand{\Seq}[1]{\langle #1\rangle}
\newcommand{\code}[1]{\ulcorner#1\urcorner}
\newcommand{\Int}[1]{\llbracket #1 \rrbracket }

\newcommand{\lamr}{{\lam^\sigma_r}}
\newcommand{\lsubst}[2]{\langle #1 := [#2]\rangle}
\newcommand{\degx}[2][x]{\mathrm{deg}_{#1}(#2)}

\newcommand{\sums}[1]{\mathscr{P}_{\mathrm{f}}(#1)}
\newcommand\pow[1]{\mathscr{P}(#1)}
\newcommand{\Set}[1]{\mathcal{#1}}
\newcommand{\Te}[1]{\mathscr{T}(#1)}
\newcommand{\Ten}[1]{\mathscr{T}^\circ(#1)}
\newcommand{\bag}[1]{{[}#1{]}}
\newcommand{\NF}[2][]{\mathrm{NF}_{#1}(#2)}

\def\eg{{\em e.g.}}
\def\ie{{\em i.e.}}
\def\cf{{\em cf.}}

\begin{document}

\title[Revisiting Call-by-value {B}\"ohm trees]{Revisiting Call-by-value {B}\"ohm trees\\ in light of their {T}aylor expansion}

\author[E.~Kerinec]{Emma Kerinec}	
\address{Universit\'e de Lyon, ENS de Lyon, Universit\'e Claude Bernard Lyon 1, LIP}	
\email{emma.kerinec@ens-lyon.fr}  

\author[G.~Manzonetto]{Giulio Manzonetto}	
\address{LIPN, UMR 7030, Universit\'e Paris 13, Sorbonne Paris Cit\'e, F-93430, Villetaneuse, France}	
\email{giulio.manzonetto@univ-paris13.fr}  

\author[M.~Pagani]{Michele Pagani}	
\address{CNRS, IRIF, Universit\'e de Paris, F-75205 Paris, France}	
\email{pagani@irif.fr}  





\begin{abstract}
  \noindent The call-by-value \lam-calculus can be endowed with permutation rules, arising from linear logic proof-nets, having the advantage of unblocking some redexes that otherwise get stuck during the reduction.
We show that such an extension allows to define a satisfying notion of B\"ohm(-like) tree and a theory of program approximation in the call-by-value setting.
We prove that all \lam-terms having the same B\"ohm tree are observationally equivalent, and characterize those B\"ohm-like trees arising as actual B\"ohm trees of \lam-terms.
 
We also compare this approach with Ehrhard's theory of program approximation based on the Taylor expansion of \lam-terms, translating each \lam-term into a possibly infinite set of so-called resource terms. 
We provide sufficient and necessary conditions for a set of resource terms in order to be the Taylor expansion of a \lam-term.
Finally, we show that the normal form of the Taylor expansion of a \lam-term can be computed by performing a normalized Taylor expansion of its B\"ohm tree. 
From this it follows that two \lam-terms have the same B\"ohm tree if and only if the normal forms of their Taylor expansions coincide.
\end{abstract}

\maketitle

\begin{quote} \itshape We are honoured to dedicate this article to Corrado B\"ohm, whose brilliant pioneering work has been an inspiration to us all.
\end{quote}

\section*{Introduction}\label{S:one}

In 1968, Corrado B\"ohm published a separability theorem -- known as \emph{the B\"ohm Theorem}~-- which is nowadays universally recognized as a fundamental theorem in \lam-calculus~\cite{boehm68}.
Inspired by this result, Barendregt in 1977 proposed the definition of ``B\"ohm tree of a \lam-term''~\cite{Bare77}, a notion which played for decades a prominent role in the theory of program approximation.
The B\"ohm tree of a \lam-term $M$ represents the evaluation of $M$ as a possibly infinite labelled tree coinductively, but effectively, constructed by collecting the stable amounts of information coming out of the computation.
Equating all \lam-terms having the same B\"ohm tree is a necessary, although non sufficient, step in the quest for fully abstract models of \lam-calculus.   
  
In 2003, Ehrhard and Regnier, motivated by insights from Linear Logic, introduced the notion of ``Taylor expansion of a \lam-term'' as an alternative way of approximating \lam-terms~\cite{EhrhardR03}.
The Taylor expansion translates a \lam-term $M$ as a possibly infinite set\footnote{%
In its original definition, the Taylor expansion is a power series of multi-linear terms taking coefficients in the semiring of non-negative rational numbers.
Following~\cite{ManzonettoP11,Ehrhard12,BoudesHP13}, in this paper we abuse language and call ``Taylor expansion'' the \emph{support} (underlying set) of the actual Taylor expansion.
This is done for good reasons, as we are interested in the usual observational equivalences between \lam-terms that overlook such coefficients. 
} of multi-linear terms, each approximating a finite part of the behaviour of $M$. 
These terms populate a \emph{resource calculus}~\cite{TranquilliTh} where \lam-calculus application is replaced by the application of a term to a bag of resources that cannot be erased, or duplicated and must be consumed during the reduction.
The advantage of the Taylor expansion is that it exposes the amount of resources needed by a \lam-term
 to produce (a finite part of) a value, a quantitative information that does not appear in its B\"ohm tree.
The relationship between these two notions of program approximation has been investigated in~\cite{EhrhardR06}, where the authors show that the Taylor expansion can actually be seen as a resource sensitive version of B\"ohm trees by demonstrating that the normal form of the Taylor expansion of $M$ is  actually equal to the Taylor expansion of its B\"ohm tree.
  
The notions of B\"ohm tree and Taylor expansion have been first developed in the setting of call-by-name (\CbN) \lam-calculus~\cite{Bare}. 
However many modern functional programming languages, like OCaml, adopt a call-by-value (\CbV) reduction strategy --- a redex of shape $(\lam x.M)N$ is only contracted when $N$ is a \emph{value}, namely a variable or a \lam-abstraction.
The~call-by-value \lam-calculus $\lam_v$ has been defined by Plotkin in 1975~\cite{Plotkin75}, but its theory of program approximation is still unsatisfactory and constitutes an ongoing line of research~\cite{Ehrhard12,CarraroG14,ManzonettoPR18}.
For instance, it is unclear what should be the B\"ohm tree of a \lam-term because of the possible presence of $\beta$-redexes that get stuck (waiting for a value) in the reduction.
A paradigmatic example of this situation is the \lam-term $M = (\lam y.\Delta)(xx)\Delta$, where $\Delta = \lam z.zz$ (see~\cite{PaoliniR99,AccattoliG17}).
This term is a call-by-value normal form because the argument $xx$, which is not a value, blocks the evaluation (while one would expect $M$ to behave as the divergent term $\Omega=\Delta\Delta$).
A significant advance in reducing the number of stuck redexes has been made in~\cite{CarraroG14} where Carraro and Guerrieri, inspired by Regnier's work in the call-by-name setting~\cite{Regnier94}, introduce permutations rules~$(\sigma)$ naturally arising from the translation of \lam-terms into Linear Logic proof-nets.
Using  $\sigma$-rules, the \lam-term $M$ above rewrites in $(\lam y.\Delta\Delta)(xx)$ which in its turn rewrites to itself, thus giving rise to an infinite reduction sequence, as desired.
In~\cite{GuerrieriPR17}, Guerrieri \emph{et al.\ }show that this extended calculus $\lsv$ still enjoys nice properties like confluence and standardization, and that adding the $\sigma$-rules preserves the operational semantics of Plotkin's \CbV\ \lam-calculus as well as the observational equivalence.

In the present paper we show that $\sigma$-rules actually open the way to provide a meaningful notion of call-by-value B\"ohm trees (Definition~\ref{def:BT}).
Rather than giving a coinductive definition, which turns out to be more complicated than expected, we follow~\cite{AmadioC98} and provide an appropriate notion of \emph{approximants}, namely \lam-terms possibly containing a constant $\bot$, that are in normal form w.r.t.\ the reduction rules of $\lsv$ (\ie, the $\sigma$-rules and the restriction of $(\beta)$ to values).
In this context, $\bot$ represents the undefined value and this intuition is reflected in the definition of a preorder $\BTle$ between approximants which is generated by  $\bot\BTle V$, for all approximated values $V$.
The next step is to associate with every \lam-term $M$ the set $\Appof M$ of its approximants and verify that they enjoy the following properties: $(i)$ the ``external shape'' of an approximant of $M$ is stable under reduction (Lemma~\ref{lem:A_preserved}); $(ii)$ two interconvertible \lam-terms share the same set of approximants (\cf, Lemma~\ref{lem:AM_preserved_by_toV}); $(iii)$ the set of approximants of $M$ is directed (Lemma~\ref{prop:AMisIdeal}).
Once this preliminary work is accomplished, it is possible to define the B\"ohm tree of $M$ as the supremum of $\Appof M$, the result being a possibly infinite labelled tree $\BT{M}$, as expected.

More generally, it is possible to define the notion of (\CbV) ``B\"ohm-like'' trees as those labelled trees that can be obtained by arbitrary superpositions of (compatible) approximants.
The B\"ohm-like trees corresponding to \CbV{} B\"ohm trees of \lam-terms have specific properties, that are due to the fact that \lam-calculus constitutes a model of computation. 
Indeed, since every \lam-term $M$ is finite, $\BT{M}$ can only contain a finite number of free variables and, since $M$ represents a program, the tree $\BT{M}$ must be computable. 
In Theorem~\ref{thm:BLTiffLamDef} we demonstrate that these conditions are actually sufficient, thus providing a characterization.

To show that our notion of B\"ohm tree is actually meaningful, we prove that all \lam-terms having the same B\"ohm tree are operationally  indistinguishable (Theorem~\ref{thm:happy_ending}) and we investigate the relationship between B\"ohm trees and Taylor expansion in the call-by-value setting.
Indeed, as explained by Ehrhard in~\cite{Ehrhard12}, the \CbV{} analogues of resource calculus and of Taylor expansion are unproblematic to define, because they are driven by solid intuitions coming from  Linear Logic: rather than using the  \CbN{} translation $A\to B = \oc A\multimap B$ of intuitionistic arrow, 
it is enough to exploit Girard's so-called ``boring'' translation, which transforms $A\to B$ in $\oc(A\multimap B)$ and is suitable for \CbV.
Following~\cite{BoudesHP13}, we define a coherence relation $\coh$ between resource terms and prove that a set of such terms corresponds to the Taylor expansion of a \lam-term if and only if it is an infinite clique  having finite height.
Subsequently, we focus on the dynamic aspects of the Taylor expansion by studying its \emph{normal form}, that can always be calculated since the resource calculus enjoys strong normalization.

In~\cite{CarraroG14}, Carraro and Guerrieri propose to extend the \CbV{} resource calculus with $\sigma$-rules  to obtain a more refined normal form of the Taylor expansion $\Te{M}$ of a \lam-term $M$ --- this allows to mimic the $\sigma$-reductions occurring in $M$ at the level of its resource approximants. 
Even with this shrewdness, it turns out that the normal form of $\Te{M}$ is different from the normal form of $\Te{\BT{M}}$, the latter containing approximants that are not normal, but whose normal form is however empty (they disappear along the reduction).
Although the result from~\cite{EhrhardR06} does not hold verbatim in \CbV, we show that it is possible to define the \emph{normalized Taylor expansion $\Ten-$ of a B\"ohm tree} and prove in Theorem~\ref{thm:T_commutes_with_BT} that the normal form of $\Te M$ coincide with $\Ten{\BT{M}}$, which is the main result of the paper.
An interesting consequence, among others, is that all denotational models satisfying the Taylor expansion (\eg, the one in~\cite{CarraroG14}) equate all \lam-terms having the same B\"ohm tree.

\subsection*{Related works} 
To our knowledge, in the literature no notion of \CbV\ B\"ohm tree appears\footnote{Even Paolini's separability result in \cite{Paolini01} for \CbV{} \lam-calculus does not rely on B\"ohm trees.}.
However, there have been attempts to develop syntactic bisimulation equivalences and theories of program approximation arising from denotational models.
Lassen~\cite{Lassen05} coinductively defines a bisimulation equating all \lam-terms having (recursively) the same ``eager normal form'', but he mentions that no obvious tree representations of the equivalence classes are at hand.
In~\cite{RonchiP04}, Ronchi della Rocca and Paolini study a filter model of \CbV{} \lam-calculus and, in order to prove an Approximation Theorem, they need to define sets of \emph{upper} and \emph{lower} approximants of a \lam-term. 
By admission of the authors~\cite{RonchiPC}, these notions are not satisfactory because they correspond to an ``over'' (resp.\ ``under'') approximation of its behaviour.

We end this section by recalling that most of the results we prove in this paper are the \CbV{} analogues of  results well-known in \CbN{} and contained in~\cite[Ch.~10]{Bare} (for B\"ohm trees), in~\cite{BoudesHP13} (for Taylor expansion) and ~\cite{EhrhardR06} (for the relationship between the two notions).

\bigskip
\noindent{\bf General notations.}
We denote by $\nat$ the set of all natural numbers.
Given a set $X$ we denote  by $\pow X$ its powerset and by $\sums X$ the set of all finite subsets of $X$.

\section{Call-By-Value $\lambda$-Calculus}
The \emph{call-by-value \lam-calculus}~$\lam_v$, introduced by Plotkin in~\cite{Plotkin75}, is a \lam-calculus endowed with a reduction relation that allows the contraction of a redex $(\lam x.M)N$ only when the argument $N$ is a value, namely when $N$ is a variable or an abstraction. 
In this section we briefly review its syntax and operational semantics.
By extending its reduction with permutation rules~$\sigma$, we obtain the calculus $\lsv$ introduced in~\cite{CarraroG14}, that will be our main subject of study.

\subsection{Its syntax and operational semantics.}\label{ssec:richer} 
For the \lam-calculus we mainly use the notions and notations from~\cite{Bare}.
We consider fixed a denumerable set $\Var$ of variables.

\begin{defi}\label{def:Lambda} The set $\Lam$ of \emph{\lam-terms} and the set $\Val$ of \emph{values} are defined through the following grammars (where $x\in\Var)$:
\[
	\begin{array}{lrl}
	(\Lam)&M,N,P,Q\bnfeq& V \mid MN\\
	(\Val)&U,V\bnfeq& x\mid \lam x.M\\
	\end{array}
\]
\end{defi}

As usual, we assume that application associates to the left and has higher precedence than \lam-abstraction. 
For instance, $\lam xyz.xyz = \lam x.(\lam y.(\lam z.((xy)z)))$.
Given $x_1,\dots,x_n\in\Var$, we let $\lam\seq x.M$ stand for $\lam x_1 \dots \lam x_n.M$. Finally, we write $MN^{\sim n}$ for $ MN\cdots N$ ($n$ times).

The set $\FV{M}$ of \emph{free variables of $M$} and the \emph{$\alpha$-conversion} are defined as in~\cite[\S2.1]{Bare}. 
A \lam-term $M$ is called \emph{closed}, or \emph{a combinator}, whenever $\FV M = \emptyset$.
The set of all combinators is denoted by $\Lam[o]$.
From now on, \lam-terms are considered up to $\alpha$-conversion, whence the symbol $=$ represents syntactic equality possibly up to renaming of bound variables.

\begin{defi} Concerning specific combinators, we define:
\[
	\begin{array}{lcllcllcl}
	\comb{I}&=&\lam x.x,&\Delta&=&\lam x.xx,&\Omega&=&\Delta\Delta,\\
	\comb{B}&=&\lam fgx.f(gx),\qquad&\comb{K}&=&\lam xy.x,\qquad&\comb{F}&=&\lam xy.y,\\
	\comb{Z}&=&\multicolumn{4}{l}{\lambda f. (\lambda y. f(\lambda z.yyz))(\lambda y. f(\lambda z.yyz)),}
	&\comb{K}^*&=&\comb{ZK}.\\
	\end{array}
\]
where $\comb{I}$ is the identity, $\Omega$ is the paradigmatic looping combinator, $\comb{B}$ is the composition operator, $\comb{K}$ and $\comb{F}$ are the first and second projection (respectively), $\comb{Z}$ is Plotkin's recursion operator, and $\comb{K}^*$ is a \lam-term producing an increasing amount of external abstractions.
\end{defi}

Given $M,N\in\Lam$ and $x\in\Var$ we denote by $M\subst{x}{N}$ the \lam-term obtained by substituting $N$ for every free occurrence of $x$ in $M$, subject to the usual proviso of renaming bound variables in $M$ to avoid capture of free variables in $N$.

\begin{rem} It is easy to check that the set $\Val$ is closed under substitution of values for free variables, namely $U,V\in\Val$ and $x\in\Var$ entail $V\subst{x}U\in\Val$.
\end{rem}

A context is a \lam-term possibly containing occurrences of a distinguished algebraic variable, called \emph{hole} and denoted by $\hole-$.
In the present paper we consider -- without loss of generality for our purposes -- contexts having a single occurrence of~$\hole-$.

\begin{defi} A \emph{(single-hole) context} $C\hole-$ is generated by the simplified grammar:
\[
	\qquad C\bnfeq \hole-\mid CM\mid MC\mid \lam x.C\qquad\textrm{(for $M\in\Lam$)}
\]
A context $C\hole-$ is called a \emph{head context} if it has shape $(\lam x_1\dots x_n.\hole-)V_1\cdots V_m$ for $V_i\in\Val$.
\end{defi}

Given $M\in\Lam$, we write $C\hole M$ for the \lam-term obtained by replacing $M$ for the hole $\hole-$ in $C\hole-$, possibly with capture of free variables.

We consider a \CbV\ \lam-calculus $\lsv$ endowed with the following notions of reductions. 
The $\beta_v$-reduction is the standard one, from~\cite{Plotkin75}, while the $\sigma$-reductions have been introduced in \cite{Regnier94,CarraroG14} and are inspired by the translation of \lam-calculus into linear logic proof-nets.

\begin{defi}\label{def:notionsofred} 
The \emph{$\beta_v$-reduction} $\to_{\beta_v}$ is the contextual closure of the following rule:
\[
	(\beta_v)\qquad	(\lam x.M)V\to M\subst{x}{V}\textrm{ whenever } V\in\Val
\]
The \emph{$\sigma$-reductions} $\to_{\sigma_1}$, $\to_{\sigma_3}$ are the contextual closures of the following rules (for $V\in\Val$):
\[
	\begin{array}{lll}
	(\sigma_1)&(\lam x.M)NP\to (\lam x.MP)N&\textrm{with }x\notin\FV{P}\\
	(\sigma_3)&V((\lam x.M)N)\to (\lam x.VM)N&\textrm{with  }x\notin\FV{V}\\
	\end{array}
\]
We also set $\to_\sigma\ =\ \to_{\sigma_1}\cup\to_{\sigma_3}$ and $\to_\V\ =\ \to_{\beta_v}\cup\to_{\sigma}$.
\end{defi}

The \lam-term at the left side of the arrow in the rule $(\beta_v)$ (resp.\ $(\sigma_1)$, $(\sigma_3)$) is called \emph{$\beta_v$-} (resp.\ $\sigma_1$-, $\sigma_3$-) \emph{redex}, while the \lam-term at the right side is the corresponding \emph{contractum}.
Notice that the condition for contracting a $\sigma_1$- (resp.\ $\sigma_3$-) redex can always be satisfied by performing appropriate $\alpha$-conversions.

Each reduction relation $\to_\rel{R}$ generates the corresponding \emph{multistep relation} $\msto[\rel{R}]$ by taking its transitive and reflexive closure,
and \emph{conversion relation} $=_\rel{R}$ by taking its transitive, reflexive and symmetric closure.
Moreover, we say that a \lam-term $M$ is in \emph{$\rel{R}$-normal form} (\emph{$\rel{R}$-nf}, for short) if there is no $N\in\Lam$ such that $M\to_{\rel{R}} N$.
We say that $M$ \emph{has an $\rel{R}$-normal form} whenever $M\msto[\rel{R}] N$ for some $N$ in $\rel R$-nf, and in this case we denote $N$ by $\nf[\rel{R}]M$.

\begin{exa}\label{exa:betav-reductions}\ 
\begin{enumerate}
\item $\comb{I}x\to_{\beta_v} x$, while $\comb{I}(xy)$ is a $\V$-normal form.
\item $\Omega\to_{\beta_v}\Omega$, whence $\Omega$ is a looping combinator in the \CbV\ setting as well.
\item $\comb{I}(\Delta(xx))$ is a $\beta_v$-nf, but contains a $\sigma_3$-redex, indeed $\comb{I}(\Delta(xx))\to_{\sigma_3}( \lam z.\comb{I}(zz))(xx)$.
\item For all values $V$, we have $\comb{Z}V =_\V V(\lam x.\comb{Z}Vx)$ with $x\notin\FV{V}$. So we get:
\item $\comb{K}^* = \comb{ZK} =_\V \comb{K}(\lam y.\comb{K}^*y) =_\V \lam x_0x_1.\comb{K}^*x_1 =_\V \lam x_0x_1x_2.\comb{K}^*x_2 =_\V\cdots =_\V \lam x_0\dots x_n.\comb{K}^*x_n$.
\item\label{exa:betav-reductions6} Let $\Xi = \comb{Z}N$ for $N = \lam f.(\lam y_1.f\comb{I})(zz)$, then we have:
\[
\begin{array}{lcll}
\Xi 	&=_\V& N(\lam w.\Xi w)&\textrm{by }\comb{Z}V =_\V V(\lam w.\comb{Z}Vw)\\
	&=_\V& (\lam y_1.(\lam w.\Xi w)\comb{I})(zz)&\textrm{by }(\beta_v)\\
	&=_\V& (\lam y_1.\Xi \comb{I})(zz)&\textrm{by }(\beta_v)\\
	&=_\V& (\lam y_1.((\lam y_2.\Xi \comb{I})(zz)) \comb{I})(zz)&\textrm{by }(\beta_v)\\
	&=_\V& (\lam y_1.((\lam y_2.\Xi \comb{I}\comb{I})(zz)))(zz)&\textrm{by ($\sigma_1$)}\\
	&=_\V& (\lam y_1.((\lam y_2.((\lam y_3.\Xi \comb{I}\comb{I}\comb{I})(zz)))(zz)))(zz)&=\ \cdots\\
\end{array}
\]
\item $\comb{ZB} =_\V \comb{B}(\lam z.\comb{ZB}z) =_\V \lam gx.(\lam z.\comb{ZB}z)(gx) =_\V
\lam gx.(\lam fy.(\lam z.\comb{ZB}z)(fy))(gx)=_\V\cdots$
\end{enumerate}
\end{exa}

The next lemma was already used implicitly in~\cite{GuerrieriPR17}.

\begin{lem}\label{lem:char_V_normal form} 
A \lam-term $M$ is in $\V$-normal form if and only if $M$ is a $G$-term generated by the following grammar (for $k\ge 0$):
\[
	\begin{array}{lcll}
	G&\bnfeq& H\mid R\\
	H&\bnfeq&x\mid \lam x.G\mid xHG_1\cdots G_k\\	
	R&\bnfeq& (\lam x.G)(yHG_1\cdots G_k)\\
	\end{array}
\]
\end{lem}
\begin{proof} $(\Rightarrow)$ Assume that $M$ is in $\V$-nf and proceed by structural induction.
Recall that every \lam-term $M$ can be uniquely written as $\lam x_1\dots x_m.M'N_1\cdots N_n$ where $m,n\ge 0$ and either $M' = x$ or $M' = (\lam x.P)Q$.  
Moreover, the \lam-terms $M',N_1,\dots, N_n$ must be in $\V$-nf's since $M$ is $\V$-nf.
Now, if $m > 0$ then $M$ is of the form $\lam x.P$ with $P$ in $\V$-nf and the result follows from the induction hypothesis.
Hence, we assume $m=0$ and split into cases depending on $M'$:
\begin{itemize}
\item
	$M'= x$ for some $x\in\Var$. If $n=0$ then we are done since $x$ is an $H$-term.
	If $n>0$ then $M = xN_1\cdots N_n$ where all the $N_i$'s are $G$-terms by induction hypothesis. Moreover, $N_1$ cannot be an $R$-term for otherwise $M$ would have a $\sigma_3$-redex.
	Whence, $N_1$ must be an $H$-term and $M$ is of the form $xHG_1\cdots G_k$ for $k=n-1$. 
\item $M' = (\lam x.P)Q$ for some variable $x$ and \lam-terms $P,Q$ in $\V$-nf.
	In this case we must have $n = 0$ because $M$ cannot have a $\sigma_1$-redex.
	By induction hypothesis, $P,Q$ are $G$-terms, but $Q$ cannot be an $R$-term or a value for otherwise $M$ would have a $\sigma_3$- or a $\beta_v$-redex, respectively.
	We conclude that the only possibility for the shape of $Q$ is $yHG_1\cdots G_k$, whence $M$ must be an $R$-term.
\end{itemize}
$(\Leftarrow)$ By induction on the grammar generating $M$. The only interesting cases are the following.
\begin{itemize}
\item $M = xHG_1\cdots G_k$ could have a $\sigma_3$-redex if $H = (\lam y.P)Q$, but this is impossible by definition of an $H$-term. 
As $H,G_1,\dots, G_k$ are in $\V$-nf by induction hypothesis, so must~be~$M$.
\item $M = (\lam x.G)(yHG_1\cdots G_k)$ where $G,H,G_1,\dots,G_k$ are in $\V$-nf by induction hypothesis. In the previous item we established that $yHG_1\cdots G_k$ is in $\V$-nf.
Thus, $M$ could only have a $\beta_v$-redex if $yHG_1\cdots G_k\in\Val$, but this is not the case by definition of $\Val$.\qedhere
\end{itemize}
\end{proof}
Intuitively, in the grammar above, $G$ stands for ``general'' normal form, $R$ for ``redex-like'' normal form and $H$ for ``head'' normal form.
The following properties are well-established.

\begin{prop}[Properties of reductions~\cite{Plotkin75,CarraroG14}]\label{prop:of_red}\ 
\begin{enumerate}
\item\label{prop:of_red1} The $\sigma$-reduction is confluent and strongly normalizing. 
\item\label{prop:of_red2} The $\beta_v$- and $\V$-reductions are confluent.
\end{enumerate}
\end{prop}

Lambda terms are classified into valuables, potentially valuable and non-potentially valuable, depending on their capability of producing a value in a suitable environment.

\begin{defi}\label{def:pot_val}
 A \lam-term $M$ is \emph{valuable} if $M\msto[\beta_v] V$ for some $V\in\Val$. 
A \lam-term $M$ is \emph{potentially valuable} if there exists a head context\footnote{%
Equivalently, $M$ is potentially valuable if there is a substitution $\vartheta : \Var\to\Val$ such that $\vartheta(M)$ is valuable.
} $C\hole- = (\lam x_1\dots x_n.\hole-)V_1\cdots V_n$, where $\FV{M} = \set{x_1,\dots, x_n}$, such that $C\hole M$ is valuable.
\end{defi}

It is easy to check that $M$ valuable entails $M$ potentially valuable and that, for $M\in\Lam[o]$, the two notions coincide.
As shown in~\cite{GuerrieriPR17}, a \lam-term $M$ is valuable (resp.\ potentially valuable) if and only if $M\msto V$ (resp.\ $C\hole M\msto V$) for some $V\in\Val$.
As a consequence, the calculus $\lsv$ can be used as a tool for studying the operational semantics of the original calculus $\lam_v$.

In~\cite{Plotkin75}, Plotkin defines an observational equivalence analogous to the following one.

\begin{defi} The \emph{observational equivalence} $\obseq$ is defined as follows (for $M,N\in\Lam$):
\[
\begin{array}{c}
	M \obseq N\\ 
	\iff\\
	\forall C\hole-\, .\, C\hole M,C\hole N\in\Lam[o] \ [\ \exists V\in\Val \,.\,C\hole M\msto[\beta_v] V \iff \exists U\in\Val\,.\,C\hole N\msto[\beta_v] U\ ]
	\end{array}
\]
\end{defi}

For example, we have $\comb{I} \obseq \lam xy.xy$ and $\Xi\equiv \Omega$ (see Example~\ref{exa:betav-reductions}\eqref{exa:betav-reductions6}), while $\Omega\not\obseq \lam x.\Omega$.

\begin{rem}\label{rem:contextlemmaholds}
It is well known that, in order to check whether $M\obseq N$ holds, it is enough to consider head contexts (cf.~\cite{LukeOng,Paolini08}). 
In other words, $M\not\obseq N$ if and only if there exists a head context $C\hole-$ such that $C\hole M$ is valuable, while $C\hole N$ is not.
\end{rem}


\section{Call-by-value B\"ohm Trees}

In the call-by-name setting there are several equivalent ways of defining B\"ohm trees.
The most famous definition is coinductive\footnote{%
See also Definition~10.1.3 of \cite{Bare}, marked by Barendregt as `informal' because at the time the coinduction principle was not as well-understood as today.
}~\cite{Lassen99}, while the formal one in Barendregt's book exploits the notion of ``effective B\"ohm-like trees'' which is not easy to handle in practice.
The definition given in Amadio and Curien's book~\cite[Def.~2.3.3]{AmadioC98} is formal, does not require coinductive techniques and, as it turns out, generalizes nicely to the \CbV{} setting.
The idea is to first define the set $\Appof M$ of approximants of a \lam-term $M$, then show that it is directed w.r.t.\ some preorder $\BTle$ and, finally, define the B\"ohm tree of $M$ as the supremum~of~$\Appof M$.

\subsection{B\"ohm trees and approximants}
Let $\Lamb$ be the set of \lam-terms possibly containing a constant $\bot$, representing the undefined \emph{value}, and let $\BTle$ be the context-closed preorder on $\Lamb$ generated by setting, for all $x\in\Var$ and $M\in\Lamb$:
\[
	\bot\BTle x,\qquad\qquad\qquad\qquad\qquad \bot \BTle \lam x.M.
\]
Notice that, by design, $\bot$ can only be used to approximate values, not \lam-terms like $\Omega$.

The reduction $\to_\V$ from Definition~\ref{def:notionsofred} generalizes to terms in $\Lamb$ in the obvious way, namely by considering a set $\Val_\bot$ of values generated by the grammar (for $M\in\Lamb$):
\[
	\begin{array}{lrl}
	(\Val_\bot)\quad& U,V\ \bnfeq& \bot\mid x \mid \lam x.M\\
	\end{array}
\]
 For example, the $\beta_v$-reduction is extended by setting for all $M,V\in\Lamb$:
\[
	(\beta_v)\qquad	(\lam x.M)V\to M\subst{x}{V}\textrm{ whenever } V\in\Val_\bot
\]
Similarly, for the $\sigma$-rules.
A \emph{$\bot$-context $C\hole-$} is a context possibly containing some occurrences of $\bot$. We use for $\bot$-contexts the same notations introduced for contexts in Section~\ref{ssec:richer}.

Given $M,N\in\Lamb$ compatible\footnote{Recall that $M,N$ are \emph{compatible} if there exists $Z$ such that $M\BTle Z$ and $N\BTle Z$.} w.r.t.\ $\BTle$, we denote their least upper bound by $M\sqcup  N$.

\begin{defi}\label{def:approx}\  
\begin{enumerate}
\item\label{def:approx1} 
The set $\App$ of \emph{approximants} contains the terms $A\in\Lamb$ generated by the grammar (for $k\ge0$):
\[
	\begin{array}{lcl}
	A &\bnfeq& B\mid C\\
	B&\bnfeq&x\mid \lam x.A\mid \bot\mid xBA_1\cdots A_k\\	
	C&\bnfeq& (\lam x.A)(yBA_1\cdots A_k) 
	\end{array}
\]
\item\label{def:approx1.5} The \emph{set of free variables} $\FV-$ is extended to approximants by setting 
$\FV{\bot} = \emptyset$.
\item\label{def:approx2} 
	Given $M\in\Lam$, \emph{the set of approximants of $M$} is defined as follows:
\[
	\AM = \set{ A\in\App\st \exists N\in\Lam, M\msto N\textrm{ and } A\BTle N}.
\]
\end{enumerate}
\end{defi}

\begin{exa}\label{ex:deuxpointdeux}
\begin{enumerate}
\item\label{ex:deuxpointdeux1} 
	$\Appof{\comb{I}} = \set{\bot,\lam x.\bot,\lam x.x}$.
\item\label{ex:deuxpointdeux2}
	$\Appof{\Omega} = \Appof{\Xi} = \emptyset$ and $\Appof{\lam x.\Omega} = \set{\bot}$.
\item\label{ex:deuxpointdeux3} 
	$\Appof{\comb{I}(\Delta(xx))} = \set{(\lam z.(\lam y.Y)(zZ))(xX) \st Y\in\set{y,\bot}\ \land\ Z\in\set{z,\bot} \ \land\ X\in\set{x,\bot}}$.\\ Notice that neither $(\lam z.\bot)(xx)$ nor $(\lam z.\bot)(x\bot)$ belong to this set, because $\bot\not\BTle \comb{I}(zz)$.
\item 
	$\Appof{\comb{Z}} = \bigcup_{n\in\nat}\set{\lam  f. f(\lam  z_0. f(\lam z_1.f\cdots(\lam  z_{n}.f\bot\, Z_{n})\cdots Z_1) Z_0) \st \forall i\,.\, Z_i\in\set{z_i,\bot}}\cup\set{\bot}$.
\item 
	$\Appof{\comb{K}^*} = \set{ \lam x_1\dots x_n.\bot \st n\ge 0}$.
\item 
	The set of approximants of $\comb{ZB}$ is particularly interesting to calculate:
	\[
	\begin{array}{r@{~}l} 
	\Appof{\comb{ZB}} =& \set{\lam f_0x_0.(\cdots(\lam f_{n-1}x_{n-1}.(\lam f_n.\bot)(f_{n-1}X_{n-1}))\cdots)(f_0X_0)\st n> 0,\forall i\,.\,X_i\in\set{x_i,\bot}}\\	
	\cup&\set{\bot,\lam f_0.\bot}.\\
	\end{array}
	\]
\end{enumerate}
\end{exa}
\begin{lem}\label{lem:Approximantsareallnormal} 
Every $M\in\App$ is in normal form with respect to the extended $\V$-reduction.
\end{lem}

\begin{proof} By a simple case analysis (analogous to the proof of Lemma~\ref{lem:char_V_normal form}).
\end{proof}

The following lemmas show that the ``external shape'' of an approximant is stable under $\V$-reduction. 
For instance, if $A = (\lam x.A_0)(yBA_1\cdots A_k)\BTle M$ then all approximants $A'\in\AM$ have shape $ (\lam x.A'_0)(yB'A'_1\cdots A'_k)$ for some $B',A'_0,\dots, A'_k\in\App$.

\begin{lem}\label{lem:technical_lem} 
Let $C\hole-$ be a (single-hole) $\bot$-context and $V\in\Val$.
Then $C\hole\bot \in\App$ and $C\hole V\to_\V N$ entails that there exists a value $V'$ such that $V\to_\V V'$ and $N = C\hole{V'}$.
\end{lem}

\begin{proof}
Let $A = C\hole\bot\in\App$.
 By Lemma~\ref{lem:Approximantsareallnormal}, $A$ cannot have any $\V$-redex. 
Clearly, substituting $V$ for an occurrence of $\bot$ in $A$ does not create any new $\beta_v$-redex, so if $C\hole V\to_{\beta_v} N$ then the contracted redex must occur in $V$. 
As $V$ is a value, it can only $\V$-reduce to a value $V'$.

It is slightly trickier to check by induction on $C\hole-$ that such an operation does not introduce any $\sigma$-redex. 
The only interesting case is $C\hole\bot = (\lam x.A')(xC'\hole\bot A_1\cdots A_k)$ where $C'\hole\bot$ is a $B$-term.
Indeed, since $x\in\Val$, $xC'\hole V$ would be a $\sigma_3$-redex for $C'\hole V = (\lam y. P)Q$ but this is impossible since $C'\hole\bot$ is a $B$-term and $B$-terms cannot have this shape.

The case $C\hole\bot = xC'\hole\bot A_1\cdots A_k$ is analogous.
\end{proof}

\begin{lem}\label{lem:A_preserved} 
For $M\in\Lam$ and $A\in\App$, $A\BTle M$ and $M\to_\V N$ entails $A\BTle N$.
\end{lem}

\begin{proof} If $A\BTle M$ then $M$ can be obtained from $A$ by substituting each occurrence of $\bot$ for the appropriate subterm of $M$, and such subterm must be a value. 
Hence, the redex contracted in $M\to_\V N$ must occur in a subterm $V$ of $M$ corresponding to an occurrence\footnote{An \emph{occurrence} of a subterm $N$ in a \lam-term $M$ is a (single-hole) context $C\hole-$ such that $M = C\hole N$.} $C\hole-$ of $\bot$ in $A$.
So we have $C\hole\bot = A$ and $C\hole V\to_\V N'$ implies, by Lemma~\ref{lem:technical_lem}, that $N'= C\hole {V'}$ for a $V'$ such that $V\to_\V V'$.
So we conclude that $A = C\hole \bot \BTle N$, as desired.
\end{proof}

\begin{lem}\label{lem:AM_preserved_by_toV} 
For $M,N\in\Lam$, $M\to_\V N$ entails $\AM = \Appof N$.
\end{lem}

\begin{proof} Straightforward from Definition~\ref{def:approx}\eqref{def:approx2} and Lemma~\ref{lem:A_preserved}.
\end{proof}

\begin{prop}\label{prop:AMisIdeal} 
For all $M\in\Lam$, the set $\AM$ is either empty or an ideal (i.e.\ non-empty, directed and downward closed) w.r.t. $\BTle$.
\end{prop}

\begin{proof} 
Assume $\AM$ is non-empty. We check the remaining two conditions:
\begin{itemize}
\item
	To show that $\AM$ is directed, we need to prove that every $A_1,A_2\in\AM$ have an upper bound $A_3\in\AM$. 
	
	We proceed by induction on $A_1$.
	In case $A_1 = \bot$ (resp.\ $A_2 = \bot$) simply take $A_3 = A_2$ (resp.\ $A_3=A_1$).
	Let us assume that $A_1,A_2 \neq \bot$.
	
	Case $A_1 = x$, then $A_3 = A_2 = x$.
	
	Case $A_1 = xB_1A^1_{1}\cdots A^1_{k}$. 
	In this case we must have $M\msto N_1$ for $N_1= xN'_0\cdots N'_{k}$ with $B_1\BTle N'_0$ and $A^1_{i}\BTle N'_i$ for all $i$ such that $1 \le i \le k$.
	As $A_2\in\AM$, there exists a \lam-term $N_2$ such that $M\msto N_2$ and $A_2\BTle N_2$. 
	By Proposition~\ref{prop:of_red}\eqref{prop:of_red2} (confluence), $N_1$ and $N_2$ have a common reduct $N$.
	Since $A_1\BTle N_1$, by Lemma~\ref{lem:A_preserved} we get $A_1\BTle N$ thus $N = xN_0\cdots N_{k}$.
	By Lemma~\ref{lem:A_preserved} again, $A_2\BTle N$ whence $A_2 = xB_2A^2_{1}\cdots A^2_{k}$ for some approximants $B_2\BTle N_0$ and $A^2_{i}\BTle N_i$ for $1\le i \le k$.
	Now, by definition, $B_1,B_2\in\Appof{N_0}$ and $A^1_{i},A^2_{i}\in \Appof{N_i}$ for $1\le i\le k$.
	By induction hypothesis, there exists $B_3\in\Appof{N_0},A^3_{i}\in\Appof{N_i}$ such that $B_1\BTle B_3\BTge B_2$ and $A^1_{i}\BTle A^3_{i}\BTge A^2_{i}$ from which it follows that the upper bound $xB_3A^3_{1}\cdots A^3_{k}$ of $A_1,A_2$ belongs to $\Appof{xN_0\cdots N_k}$.
	By Lemma~\ref{lem:AM_preserved_by_toV}, we conclude that $xB_3A^3_{1}\cdots A^3_{k}\in\AM$, as desired.

	Case $A_1 = (\lam x.A'_1)(yB_1A^1_{1}\cdots A^1_{k})$. 
	In this case we must have $M\msto N_1$ for $N_1=(\lam x.M')(yM_0\cdots M_k)$ with $A'_1\BTle M'$, $B_1\BTle M_0$ and $A^1_{i}\BTle M_i$ for all $i$ such that $1 \le i \le k$.
	Reasoning as above, $A_2\in\AM$ implies there exists a \lam-term $N_2$ such that $M\msto N_2$ and $A_2\BTle N_2$. 
	By Proposition~\ref{prop:of_red}\eqref{prop:of_red2}, $N_1$ and $N_2$ have a common reduct $N$.
	Since $A_1\BTle N_1$, by Lemma~\ref{lem:A_preserved} we get $A_1\BTle N$ thus $N = (\lam x.N')(yN_0\cdots N_k)$.
	By Lemma~\ref{lem:A_preserved} again, $A_2\BTle N$ whence $A_2 = (\lam x.A'_2)(yB_2A^2_{1}\cdots A^2_{k})$ where $A'_2\BTle N'$, $B_2\BTle N_0$ and  $A^2_{i}\BTle N_i$ for $1 \le i\le k$.
	By induction hypothesis we get 
	$A'_3\in\Appof{N'}$ such that $A'_1\BTle A'_3 \BTge A'_2$, 
	$B_3\in\Appof{N_0}$ such that $B_1\BTle B_3 \BTge B_2$ and 
	$A^3_{i}\in\Appof{N_i}$ such that $A^1_{i}\BTle A^3_{i}\BTge A^2_{i}$ for $1\le i\le k$. 
	It follows that the upper bound $(\lam x.A'_3)(yB_3A^3_{1}\cdots A^3_{k})$ of $A_1,A_2$ belongs to $\Appof{(\lam x.N')(yN_0\cdots N_k)}$.
		By Lemma~\ref{lem:AM_preserved_by_toV}, we conclude that $(\lam x.A'_3)(yB_3A^3_{1}\cdots A^3_{k})\in\AM$.
	
	All other cases follow from Lemma~\ref{lem:A_preserved}, confluence of $\to_\V$ and the induction hypothesis.
\item To prove that $\Appof M$ is downward closed, we need to show that for all $A_1,A_2\in\App$, if $A_1\BTle A_2\in\Appof M$ then $A_1\in\Appof M$, but this follows directly from its definition.\qedhere
\end{itemize}
\end{proof}

As a consequence, whenever $\AM\neq\emptyset$, we can actually define the B\"ohm tree of a \lam-term $M$ as the supremum of its approximants in $\Appof M$.

\begin{defi}\label{def:BT} 
\begin{enumerate}
\item\label{def:BT1}
Let $M\in\Lam.$
The \emph{(call-by-value) B\"ohm tree of $M$}, in symbols $\BT{M}$, is defined as follows (where we assume that $\Sup\emptyset = \emptyset$):
\[
	\BT M = \Sup \AM
\]
Therefore, the resulting structure is a possibly infinite labelled tree $T$.
\item\label{def:BT2}
 More generally, every $\Set X\subseteq \App$ directed and downward closed determines a so-called \emph{B\"ohm-like tree} $T = \Sup \Set X$. 
\item Given a B\"ohm-like tree $T$, we set $\FV T =  \FV{\Set X} = \bigcup_{A \in\Set X} \FV A$.
\end{enumerate}
\end{defi}

The difference between the B\"ohm tree of a \lam-term $M$ and a B\"ohm-like tree $T$ is that the former must be ``computable\footnote{The formal meaning of ``computable'' will be discussed in the rest of the section.}''  since it is \lam-definable, while the latter can be arbitrary. 
In particular, any B\"ohm tree $\BT{M}$ is a B\"ohm-like tree but the converse does not hold.

\begin{rem}\label{rem:Appof_iff_BT} 
\begin{enumerate}
\item 
Notice that $\AM = \Appof N$ if and only if $\BT{M} = \BT{N}$.
\item The supremum $\bigsqcup\Set X$ in Definition~\ref{def:BT}\eqref{def:BT2} (and a fortiori $\bigsqcup\Appof{M}$, in~\eqref{def:BT1}) belongs to the larger set $\mathscr{X}$ generated by taking the grammar in Definition~\ref{def:approx}\eqref{def:approx1} coinductively, whose elements are ordered by $\BTle$ extended to infinite terms. However, $\mathscr{X}$ contains terms like
\[
(\lam y_1.(\lam y_2.(\lam y_3.\cdots)(yy))(yy))(yy)\in\mathscr{X}
\]
that are not ``B\"ohm-like'' as they cannot be obtained as the supremum of a directed subset $\Set X\subseteq \App$.
\item 
$\FV{\BT{M}} \subseteq \FV{M}$ and the inclusion can be strict: $\FV{\BT{\lam x.\Omega y}} =\FV{\bot} = \emptyset$.
\end{enumerate}
\end{rem}
 The B\"ohm-like trees defined above as the supremum of a set of approximants can be represented as actual trees. 
Indeed, any B\"ohm-like tree $T$ can be depicted using the following ``building blocks''.
\begin{itemize}
\item If $T = \bot$ we actually draw a node labelled $\bot$.
\item If $T = \lam x.T'$ we use an \emph{abstraction node} labelled ``$\lam x$'':
	\begin{center}
	\begin{tikzpicture}
	\node (root) {$\lam x$};
	\node (T) at (0,-1) {$T'$};
	\draw (root) -- (T);
	\end{tikzpicture}
	\end{center}
\item If $T = xT_1 \cdots T_k$, we use an \emph{application node} labelled by ``$@$'':
	\begin{center}
	\begin{tikzpicture}
	\node (root) {$@$};
	\node (T) at (-8pt,-1) {$x\quad T_1\quad \cdots\quad T_k$};
	\draw (root) -- ($(T)+(30pt,10pt)$);
	\draw (root) -- ($(T)+(-34pt,10pt)$);	
	\draw (root) -- ($(T)+(-18pt,10pt)$);		
	\end{tikzpicture}
	\end{center}
\item If $T = (\lam x.T_0)(yT_1\cdots T_k)$ we combine the application and abstraction nodes as imagined:
	\begin{center}
	\begin{tikzpicture}
	\node (root) at (-1,1) {$@$};
	\node (lamx) at (-2.5,2pt) {$\lam x$};	
	\node (T0) at (-2.5,-1) {$T_0$};
	\node (y) at (0.5,0) {$@$};
	\node (T) at (0.2,-1) {$y\quad T_1\quad \cdots\quad T_k$};
	\draw (y) -- ($(T)+(30pt,10pt)$);
	\draw (y) -- ($(T)+(-35pt,10pt)$);	
	\draw (y) -- ($(T)+(-18pt,10pt)$);
	\draw (root) -- (y);
	\draw (root) -- (lamx);	
	\draw (lamx) -- (T0);
	\end{tikzpicture}
	\end{center}
\end{itemize}
Notice that the tree $T_1$ in the last two cases need to respect the shape of the corresponding approximant (Definition~\ref{def:approx}\eqref{def:approx1}) for otherwise $T$ would not be the supremum of an ideal.

\begin{exa} Notable examples of B\"ohm trees of \lam-terms are given in Figure~\ref{fig:BTs}.
Interestingly, the \lam-term $\Xi$ from Example~\ref{exa:betav-reductions}\eqref{exa:betav-reductions6}   satisfying
\begin{equation}\label{eq:Xi}
	\Xi =_\V (\lam y_1.((\lam y_2.(\cdots
	(\lam y_n.\Xi \comb{I}^{\sim n})(zz)
	\cdots))(zz)))(zz)
\end{equation}
is such that $\BT{\Xi} = \bot$. Indeed, substituting $\bot$ for a $\lam y_n.\Xi \comb{I}^{\sim n}$ in \eqref{eq:Xi} never gives an approximant belonging to $\App$ (\cf\ the grammar of Definition~\ref{def:approx}\eqref{def:approx1}).
\begin{figure}[t!]
\begin{tikzpicture}
\node (root) {};
\node (BTo) at (-10,0) {$\BT{\Omega}$};
\node (eq) at ($(BTo)+(0,-12pt)$) {$\shortparallel$};
\node at ($(eq)+(0,-12pt)$) {$\emptyset$};
\node (BTlxo) at ($(BTo)+(75pt,0pt)$) {$\BT{\lam x.\Omega}$};
\node (eq) at ($(BTlxo)+(0,-12pt)$) {$\shortparallel$};
\node (lx) at ($(eq)+(0,-12pt)$) {$\bot$};
\node (BTZ) at ($(BTlxo)+(160pt,0pt)$) {$\BT{\comb{Z}}$};
\node (eq) at ($(BTZ)+(0,-12pt)$) {$\shortparallel$};
\node (lf) at ($(eq)+(0,-12pt)$) {$\lam f$};
\node (at1) at ($(lf)+(0,-30pt)$) {$@$};
\draw (lf) -- (at1);
\node (f1) at ($(at1)+(-30pt,-30pt)$) {$f$};
\node (lz0) at ($(at1)+(30pt,-30pt)$) {$\lam z_0$};
\draw (at1) -- (f1);
\draw (at1) -- (lz0);
\node (at2) at ($(lz0)+(0,-30pt)$) {$@$};
\node (f2) at ($(at2)+(-40pt,-30pt)$) {$f$};
\node (lz1) at ($(at2)+(0pt,-30pt)$) {$\lam z_1$};
\node (z0) at ($(at2)+(40pt,-31pt)$) {$z_0$};
\draw (lz0) -- (at2);
\draw (at2) -- (f2);
\draw (at2) -- (lz1);
\draw (at2) -- (z0);
\node (at3) at ($(lz1)+(0,-30pt)$) {$@$};
\node (f2) at ($(at3)+(-40pt,-30pt)$) {$f$};
\node (lz2) at ($(at3)+(0pt,-30pt)$) {$\lam z_2$};
\node (z1) at ($(at3)+(40pt,-32pt)$) {$z_1$};
\draw (lz1) -- (at3);
\draw (at3) -- (f2);
\draw (at3) -- (lz2);
\draw (at3) -- (z1);
\draw[dotted] (lz2) -- ($(lz2)+(0,-20pt)$);
\node (BTBoh) at ($(BTo)+(20pt,-85pt)$) {$\BT{\comb{I}(zz)}$};
\node (eq) at ($(BTBoh)+(0,-15pt)$) {$\shortparallel$};
	\node (hi) at ($(eq)+(0,-15pt)$) {$@$};
	\node (lamx) at ($(hi)+(-30pt,-30pt)$) {$\lam x$};	
	\node (x) at ($(lamx)+(0,-30pt)$) {$x$};
	\node (y) at ($(hi)+(30pt,-30pt)$) {$@$};
	\node (T1) at ($(y)+(15pt,-30pt)$) {$z$};
	\node (T2) at ($(y)+(-15pt,-30pt)$) {$z$};	
	\draw (y) -- (T1);	
	\draw (y) -- (T2);		
	\draw (hi) -- (y);
	\draw (hi) -- (lamx);	
	\draw (lamx) -- (x);
\node (BTKst) at ($(BTZ)+(100pt,0pt)$) {$\BT{\comb{\comb{K}^*}}$};
\node (eq) at ($(BTKst)+(0,-12pt)$) {$\shortparallel$};	
\node (lx0) at ($(eq)+(0,-12pt)$) {$\lam x_0$};
\node (lx1) at ($(lx0)+(0,-30pt)$) {$\lam x_1$};
\node (lx2) at ($(lx1)+(0,-30pt)$) {$\lam x_2$};
\node (lx3) at ($(lx2)+(0,-30pt)$) {$\lam x_3$};
\draw (lx1) -- (lx0);
\draw (lx2) -- (lx1);
\draw (lx3) -- (lx2);
\draw[dotted] (lx3) -- ($(lx3) + (0,-20pt)$);
\node (BTZB) at ($(BTo)+(145pt,0pt)$) {$\BT{\comb{ZB}}$};
\node (eq) at ($(BTZB)+(0,-12pt)$) {$\shortparallel$};	
\node (lf0) at ($(eq)+(0,-12pt)$) {$\lam f_0$};
\node (lx0) at ($(lf0)+(0,-30pt)$) {$\lam x_0$};
\node (at0) at ($(lx0)+(0,-30pt)$) {$@$};
\node (at01) at ($(at0)+(30pt,-30pt)$) {$@$};
\node (f0) at ($(at01)+(-15pt,-30pt)$) {$f_0$};
\node (x0) at ($(at01)+(15pt,-31pt)$) {$x_0$};
\node (lf1) at ($(at0)+(-30pt,-30pt)$) {$\lam f_1$};
\node (lx1) at ($(lf1)+(0pt,-30pt)$) {$\lam x_1$};
\node (at1) at ($(lx1)+(0,-30pt)$) {$@$};
\node (at11) at ($(at1)+(30pt,-30pt)$) {$@$};
\node (f1) at ($(at11)+(-15pt,-30pt)$) {$f_1$};
\node (x1) at ($(at11)+(15pt,-31.5pt)$) {$x_1$};
\node (lf2) at ($(at1)+(-30pt,-30pt)$) {$\lam f_2$};
\node (lx2) at ($(lf2)+(0pt,-30pt)$) {$\lam x_2$};
\draw (lf0) -- (lx0);
\draw (at0) -- (lx0);
\draw (at0) -- (at01);
\draw (at0) -- (lf1);
\draw (lf1) -- (lx1);
\draw (f0) -- (at01);
\draw (x0) -- (at01);
\draw (at1) -- (lx1);
\draw (lf2) -- (at1);
\draw (at11) -- (at1);
\draw (lf2) -- (lx2);
\draw (at11) -- (f1);
\draw (at11) -- (x1);
\draw[dotted] (lx2) -- ($(lx2)+(0,-20pt)$);
\end{tikzpicture}
\caption{Examples of \CbV\ B\"ohm trees.}\label{fig:BTs}
\end{figure}

\end{exa}
\begin{samepage}
\begin{prop}\label{prop:M=VNimpBTM=BTN} 
For $M,N\in\Lam$, if $M =_\V N$ then $\BT{M} = \BT{N}$.
\end{prop}
\begin{proof}
By Proposition~\ref{prop:of_red}\eqref{prop:of_red2} (\ie\ confluence of $\to_\V$), $M =_\V N$ if and only if there exists a \lam-term $P$ such that $M\msto[\V] P$ and $N\msto[\V]P$.
By an iterated application of Lemma~\ref{lem:AM_preserved_by_toV} we get $\Appof{M} = \Appof{P}= \Appof{N}$, so we conclude $\BT{M} = \BT{N}$.
\end{proof}
\end{samepage}

Theorem~\ref{thm:BLTiffLamDef} below provides a characterization of those B\"ohm-like trees arising as the B\"ohm tree of some \lam-term, in the spirit of~\cite[Thm.~10.1.23]{Bare}.
To achieve this result, it will be convenient to consider a tree as a set of sequences closed under prefix.

We denote by $\nat^*$ the set of finite sequences of natural numbers.
Given $n_1,\dots,n_k\in\nat$, the corresponding sequence $\sigma\in\nat^*$ of length $k$ is represented by $\sigma=\Seq{n_1,\dots,n_k}$. 
In particular, $\Seq{}$ represents the empty sequence of length 0.
Given $\sigma\in\nat^*$ as above and $n\in\nat$, we write $n :: \sigma$ for the sequence $\Seq{n,n_1,\dots,n_k}$ and $\sigma;n$ for the sequence $\Seq{n_1,\dots,n_k,n}$.

Given a tree $T$, the sequence $i :: \sigma$ possibly determines a subtree that can be found going through the $(i+1)$-th children of $T$ (if it exists) and then following the path $\sigma$. Of course this is only the case if $i :: \sigma$ actually belongs to the domain of the tree.
The following definition formalizes this intuitive idea in the particular case of syntax trees of approximants.

\begin{defi} Let $\sigma\in\nat^*,A\in\App$. The \emph{subterm of $A$ at $\sigma$}, written $A_\sigma$, is defined by:
\[
	\begin{array}{lrcl}
	A_{\Seq{}} =A &(\lam x.A)_\sigma &=&\begin{cases}
					A_\tau&\textrm{if }\sigma = 0::\tau,\\
					\uparrow&\textrm{otherwise},
					\end{cases}\\~\\
	\bot_\sigma =\ \uparrow&
	(xA_0\cdots A_k)_\sigma &=&
					\begin{cases}
					(A_{i-1})_\tau&\textrm{if $1\le i\le k+1$ and }\sigma = i :: \tau,\\
					\uparrow&\textrm{otherwise},
					\end{cases}\\~\\
	&
	((\lam x.A')(yA_0\cdots A_k))_\sigma &=& \begin{cases}
					A'_\tau&\textrm{if }\sigma = 0::0::\tau,\\
					(A_{i-1})_\tau&\textrm{if $1\le i\le k+1$ and }\sigma = 1:: i :: \tau,\\		
					\uparrow&\textrm{otherwise.}
	\end{cases}\\
	\end{array}
\]
As a matter of notation, given an approximant $A'$, a subset $\Set X\subseteq\App$ and a sequence $\sigma\in\nat^*$, we write $\exists A_\sigma \simeq_{\Set X} A'$ whenever there exists $A\in\Set X$ such that $A_\sigma$ is defined and $A_\sigma = A'$.
\end{defi}

\begin{thm}\label{thm:BLTiffLamDef} 
Let $\Set X\subseteq \App$ be a set of approximants.
There exists $M\in\Lam$ such that $\AM = \Set X$ if and only if the following three conditions hold:
\begin{enumerate}
\item\label{thm:BLTiffLamDef1} 
	$\Set X$ is directed and downward closed w.r.t.\ $\BTle$,
\item\label{thm:BLTiffLamDef2} 
	$\Set X$ is r.e.\ (after coding),	
\item\label{thm:BLTiffLamDef3} 
	$\FV{\Set{X}}$ is finite.
\end{enumerate} 
\end{thm}

\begin{proof}[Proof sketch]
$(\Rightarrow)$ Let $M\in\Lam$ be such that $\Set X = \Appof M$, then \eqref{thm:BLTiffLamDef1} is satisfied by Proposition~\ref{prop:AMisIdeal} and \eqref{thm:BLTiffLamDef3} by Remark~\ref{rem:Appof_iff_BT}.
Concerning \eqref{thm:BLTiffLamDef2}, let us fix an effective bijective encoding $\# : \Lamb\to \nat$.
Then the set $\set{ \# A\st A \in \Set X}$ is r.e.\ because it is semi-decidable to determine if $M\msto N$ (just enumerate all $\V$-reducts of $M$ and check whether $N$ is one of them),
the set $\set{\#A\st A\in \App}$ and the relation $\BTle$ restricted to $\App\times\Lam$ are decidable.

$(\Leftarrow)$
Assume that $\Set X$ is a set of approximants satisfying the conditions (\ref{thm:BLTiffLamDef1}-\ref{thm:BLTiffLamDef3}). 

If $\Set X = \emptyset$ then we can simply take $M = \Omega$ since $\Appof{\Omega} = \emptyset$.

If $\Set X$ is non-empty then it is an ideal.
Since $\Set X$ is r.e., if $A'\in\App$ and $\sigma\in\nat^*$ are effectively given then the condition $\exists A_\sigma \simeq_{\Set X} A'$ is semi-decidable and a witness $\comb{A}$ can be computed.
Let $\code{\sigma}$ be the numeral associated with $\sigma$ under an effective encoding and $\code{A}$ be the \emph{quote of $A$} as defined by Mogensen\footnote{%
This encoding is particularly convenient because it is effective, defined on open terms by exploiting the fact that $\FV{\code{M}} = \FV{M}$ and works in the \CbV{} setting as well (easy to check).
See also~\cite[\S6.1]{Bare2} for a nice treatment.
}
 in~\cite{Mogensen92}, using a fresh variable $z_b\notin\FV{\Set X}$ to represent the $\bot$. 
(Such variable always exists because $\FV{\Set X}$ is finite.)
The \CbV{} \lam-calculus being Turing-complete, as shown by Paolini in~\cite{Paolini01}, there exists a \lam-term $P_{\Set X}$ satisfying:
\[
	P_{\Set X}\code{\sigma}\code{A'} =_\V \begin{cases}
	\code{\comb{A}}&\textrm{ if $\exists A_\sigma \simeq_{\Set X} A'$ holds},\\
	\textrm{not potentially valuable}&\textrm{ otherwise.}\\
	\end{cases}
\]
for some witness $\comb{A}$.
Recall that there exists an \emph{evaluator} $\comb{E}\in\Lambda^o$ such that $\comb{E}\code{M} =_\V M$ for all \lam-terms $M$.
Using the \lam-terms $\comb{E},P_{\Set X}$ so-defined and the recursion operator $\comb{Z}$, it is possible to define a \lam-term $F$ (also depending on $\Set X$) satisfying the following recursive equations:
\[
F\,\code{\sigma} =_\V\!\! \begin{cases}
				  x &\textrm{if } \exists A_\sigma \simeq_{\Set X} x,\\
				  \lam x.F\,\code{\sigma; 0}&\textrm{if }\exists A_\sigma \simeq_{\Set X} \lam x.A_1,\\
				  x(F\,\code{\sigma;1})\cdots (F\,\code{\sigma;k+1})&\textrm{if }\exists A_\sigma \simeq_{\Set X} xA_0\cdots A_k,\\
				  \big(\lam x.F\code{\sigma;0;0}\big)\big(y(F\code{\sigma;1;1})\cdots(F\code{\sigma;1;k+1})\big)&\textrm{if }\exists A_\sigma \simeq_{\Set X} (\lam x.A)(yA_0\cdots A_k),\\
				  \textrm{not valuable}&\textrm{otherwise.}\\
				  \end{cases}
\]
The fact that $\Set X$ is directed guarantees that, for a given sequence $\sigma$, exactly one of the cases above is applicable.
It is now easy to check that $\Appof{F\,\code{\Seq{}}} = \Set X$.
\end{proof}


\section{Call-By-Value Taylor Expansion}\label{Sec:TaylorExp}
The \emph{(call-by-name) resource calculus} $\lam_r$ has been introduced by Tranquilli in his thesis~\cite{TranquilliTh}, and its promotion-free fragment is the target language of Ehrhard and Regnier's Taylor expansion~\cite{EhrhardR06}.
Both the resource calculus and the notion of Taylor expansion have been adapted to the \CbV\ setting by Ehrhard~\cite{Ehrhard12}, using Girard's second translation of intuitionistic arrow in linear logic. 
Carraro and Guerrieri added to \CbV\ $\lam_r$ the analogous of the $\sigma$-rules and studied the denotational and operational properties of the resulting language $\lamr$ in~\cite{CarraroG14}.

\subsection{Its syntax and operational semantics.} 
We briefly recall here the definition of the \emph{call-by-value resource calculus} $\lamr$ from~\cite{CarraroG14}, and introduce some notations.

\begin{defi} The sets $\Val[r]$ of \emph{resource values}, $\Lam[s]$ of \emph{simple terms} and $\Lam[r]$ of \emph{resource terms} are generated by the following grammars (for $k\ge 0$):
\[
	\begin{array}{llcll}
	(\Val[r])&u,v&\bnfeq&x\mid\lam x.t&\textrm{resource values}\\
	(\Lam[s])&s,t&\bnfeq&st\mid \bag{v_1,\dots,v_k}&\textrm{simple terms}\\
	(\Lam[r])&e&\bnfeq& v\mid s&\textrm{resource terms}\\
	\end{array}
\]
The notions of \emph{$\alpha$-conversion} and \emph{free variable} are inherited from $\lsv$.
In particular, given $e\in\Lam[r]$, $\FV{e}$ denotes the set of free variables of $e$.
The \emph{size} of a resource term $e$ is defined in the obvious way, while the \emph{height $\hgt e$ of $e$} is the height of its syntax tree: 
\[
	\begin{array}{lcl}
	\hgt x &=& 0,\\
	\hgt {\lam x.t} &=& \hgt t + 1,\\
	\hgt {st} &=& \max\set{\hgt s, \hgt t} + 1,\\
	\hgt{\bag{v_1,\dots,v_k}} &=& \max\set {\hgt{v_i}\st i \le k} + 1.\\
	\end{array}
\]
\end{defi}

Resource values are analogous to the values of $\lsv$, namely variables and \lam-abstractions.
Simple terms of shape $\bag{v_1,\dots,v_n}$ are called \emph{bags} and represent finite multisets of linear resources --- this means that every $v_i$ must be used exactly once along the reduction.
Indeed, when a singleton bag $\bag{\lam x.t}$ is applied to a bag $\bag{v_1,\dots,v_n}$ of resource values, each $v_i$ is substituted for exactly one free occurrence of $x$ in $t$. 
Such an occurrence is chosen non-deterministically, and all possibilities are taken into account --- this is expressed by a set-theoretical union of resource terms (see Example~\ref{ex:reductions} below).
In case there is a mismatch between the cardinality of the bag and the number of occurrences of $x$ in $t$, the reduction relation ``raises an exception'' and the result of the computation is the empty set $\emptyset$.

Whence, we need to introduce some notations concerning sets of resource terms.

\begin{nota}\label{notation:sumsofresterm}
Sets of resource values, simple terms and resource terms are denoted by: 
\[
	\Set{U},\Set{V}\in\pow{\Val[r]},\quad \Set{S},\Set{T}\in\pow{\Lam[s]},\quad \Set{E}\in\pow{\Lam[r]},
\]
To simplify the subsequent definitions, given $\Set{S,T}\in\pow{\Lam[s]}$ and $\Set{V}_1,\dots,\Set{V}_k\in\pow{\Val[r]}$ we fix the following notations (as a syntactic sugar, not as actual syntax):
\[
	\begin{array}{ccl}
	\lam x.\Set{T} &=& \set{ \lam x.t \st t\in\Set{T}}\in\pow{\Val[r]},\\
	\Set{S\,T} &=&\{ st \st s\in\Set{S}, t\in\Set{T}\}\in\pow{\Lam[s]},\\
	\bag{\Set{V}_1,\dots,\Set{V}_k} &=& \set{\bag{v_1,\dots,v_k}\st v_1\in\Set{V}_1,\dots,v_k\in\Set{V}_k}\in\pow{\Lam[s]}.\\
	\end{array}
\]
Indeed all constructors of $\lamr$ are multi-linear, so we get $\lam x.\emptyset = \emptyset\Set{T}= \Set{S}\emptyset = [\emptyset,\Set{V}_1,\dots,\Set{V}_k] = \emptyset$.
\end{nota}

These notations are used in a crucial way, \eg, in Definition~\ref{def:Lamr:reds}\eqref{def:Lamr:reds2}.

\begin{defi} Let $e\in\Lam[r]$ and $x\in\Var$.
\begin{enumerate}
\item 
	Define the \emph{degree of $x$ in $e$}, written $\degx e$, as the number of  free occurrences of the variable $x$ in the resource term $e$. 
\item 
	Let $e\in\Lam[r]$, $v_1,\dots,v_n\in\Val[r]$ and $x\in\Var$.
	The \emph{linear substitution of $v_1,\dots,v_n$ for $x$ in $e$}, denoted by $e\lsubst{x}{v_1,\dots,v_n}\in\sums{\Lam[r]}$, is defined as follows:
\[
	e\lsubst{x}{v_1,\dots,v_n} = 
	\begin{cases}
	\big\{ e[x_1:=v_{\sigma(1)},\dots,x_n:=v_{\sigma(n)}] \st \sigma\in\mathfrak{S}_n\big\},&\textrm{if }\degx e = n,\\
	\emptyset,&\textrm{otherwise.}
	\end{cases}
\]
where $\mathfrak{S}_n$ is the group of permutations over $\{1,\dots,n\}$ and $x_1,\dots,x_n$ is an enumeration of the free occurrences of $x$ in $e$, so that $e\subst{x_i}{v_{\sigma(i)}}$ denotes the resource term obtained from $e$ by replacing the $i$-th free occurrence of $x$ in $e$ with the resource value $v_{\sigma(i)}$.
\end{enumerate}
\end{defi}

The definitions above open the way to introduce the following notions of reduction for~$\lamr$,
mimicking the corresponding reductions of $\lsv$ (\cf\ Definition~\ref{def:notionsofred}).

\begin{defi}\label{def:Lamr:reds}
\begin{enumerate}
\item\label{def:Lamr:reds1} 
	The \emph{$\beta_r$-reduction} is a relation $\to_{\beta_r}\,\subseteq \Lam[r]\times\sums{\Lam[r]}$ defined by the following rule (for $v_1,\dots,v_n\in\Val[r]$):
\[
	(\beta_r)\quad [\lam x.t][v_1,\dots,v_k] \to t\lsubst{x}{v_1,\dots,v_n}.
\]
Similarly, the \emph{$0$-reduction} $\to_0\,\subseteq \Lam[r]\times\sums{\Lam[r]}$ is defined by the rule:
\[
	(0)\quad [v_1,\dots,v_n]\,t \to \emptyset,\quad\textrm{ when }n\neq 1.\qquad
\]
The \emph{$\sigma$-reductions} $\to_{\sigma_1},\to_{\sigma_3}\,\subseteq \Lam[r]\times\Lam[r]$ are defined by the rules:
\[
	\begin{array}{lll}
	(\sigma_1)&[\lam x.t]s_1s_2\to [\lam x.ts_2]s_1,&\textrm{if }x\notin\FV{s_1},\\
	(\sigma_3)&[v]([\lam x.t]s)\to [\lam x.[v]t]s,&\textrm{if }x\notin\FV{v}\textrm{ and }v\in\Val[r].\\
	\end{array}
\]
\item\label{def:Lamr:reds2} The relation $\to_\R\ \subseteq\sums{\Lam[r]}\times\sums{\Lam[r]}$ is the contextual closure of the rules above, \ie\ $\to_\R$ is the smallest relation including $(\beta_r), (0), (\sigma_1), (\sigma_3)$ and satisfying the rules in Figure~\ref{fig:CTXTR}.
\begin{figure}[tt]
\[
\begin{array}{c}
\hspace{10pt}\infer{\lam x.t\to_\R\lam x.\Set{T}}{t\to_\R\Set{T}}\qquad
\infer{s\,t \to_\R\Set{S}\,t}{s\to_\R\Set{S}}\qquad
\infer{s\,t\to_\R s\,\Set{T}}{t\to_\R\Set{T}}\qquad
\infer{\bag{v_0,v_1,\dots,v_k}\to_\R \bag{\Set V_0,v_1,\dots,v_k}}{v_0\to_\R \Set V_0}\\[2ex]
\infer{\set{e} \cup\Set E_2\to_\R\Set E_1 \cup \Set E_2}{e\to_\R\Set E_1&e\notin\Set E_2}\\
\end{array}
\]
\caption{Contextual rules for $\to_\R\ \subseteq\sums{\Lam[r]}\times\sums{\Lam[r]}$.}\label{fig:CTXTR}
\end{figure}
\item\label{def:Lamr:reds3} The transitive and reflexive closure of $\to_\R$ is denoted by $\msto[\R]$, as usual. 
\end{enumerate}
\end{defi}

\begin{exa}\label{ex:reductions} We provide some examples of reductions:
\begin{enumerate}
\item $\bag{\lam x.\bag{x}\bag{x}}\bag{\lam y.\bag{y},z}\to_{\beta_r} 
	\set{ \bag{\lam y.\bag y}\bag{z}, \bag{z}\bag{\lam y.\bag y} }\to_{\beta_r} 
	\set{ \bag{z}, \bag{z}\bag{\lam y.\bag y}}$.
\item $\bag{\lam x.\bag{x,x}}\bag{\lam y.\bag{y},z}\to_{\beta_r} 
	\set{\bag{\lam y.\bag{y},z},\bag{z,\lam y.\bag{y}}} = \set{\bag{\lam y.\bag{y},z}$}.
\item $\bag{\lam y.\bag{\lam x.\bag{x,x}\bag{y}}}(\bag{z}\bag{w})\bag{\comb{I},w}\to_{\sigma_1}
	\bag{\lam y.\bag{\lam x.\bag{x,x}\bag{y}}\bag{\comb{I},w}}(\bag{z}\bag{w})
	\to_{\beta_r}
	\set{\bag{\lam y.\bag{\comb{I},w}\bag{y}}(\bag{z}\bag{w})}
	$ $\to_0 \emptyset$. 
	Note that $(\sigma_1)$ is used to unblock an otherwise stuck $\beta_r$-redex.
\item\label{ex:same4} $\bag{\comb{I}}(\bag{\lam x.\bag{\lam y.\bag{x}\bag{y} }}\bag{z}\bag{w})\to_{\beta_r}
	\set{
	\bag{\lam x.\bag{\lam y.\bag{x}\bag{y} }}\bag{z}\bag{w}
	}\to_{\beta_r}
	\set{\bag{\lam y.\bag{z}\bag{y} }\bag{w}}
	\to_{\beta_r}
	\set{\bag{z}\bag{w}}
	$. 
\item\label{ex:same5}	$\bag{\comb{I}}(\bag{\lam x.\bag{\lam y.\bag{x}\bag{y} }}\bag{z}\bag{w}) \to_{\sigma_3}
	\bag{\lam x.\bag{\comb{I}}\bag{\lam y.\bag{x}\bag{y} }}\bag{z}\bag{w}\to_{\sigma_1}
	\bag{\lam x.\bag{\comb{I}}\bag{\lam y.\bag{x}\bag{y} }\bag{w}}\bag{z}\to_{\beta_r}
	\set{\bag{\comb{I}}\bag{\lam y.\bag{z}\bag{y} }\bag{w}}
	\to_{\beta_r}
	\set{\bag{\lam y.\bag{z}\bag{y}}\bag{w}}
	\to_{\beta_r}
	\set{\bag{z}\bag{w}}
	$.
\end{enumerate}
Remark that \eqref{ex:same4} and \eqref{ex:same5} constitute two different reduction sequences originating from the same simple term.
\end{exa}

As shown in~\cite{CarraroG14}, this notion of reduction enjoys the following properties.

\begin{prop}\label{prop:confluence_Lamr} 
The reduction $\to_\R$ is confluent and strongly normalizing.
\end{prop}

Note that strong normalization is straightforward to prove ---
indeed, a $0$-reduction annihilates the whole term, $\sigma$-rules are strongly normalizing (\cf\ Proposition~\ref{prop:of_red}\eqref{prop:of_red1}) and contracting a $\beta_r$-redex in a resource term $e$ produces a set of resource terms whose size is strictly smaller than $e$ because no duplication is involved and a $\lam$-abstraction is erased.

As a consequence of Proposition~\ref{prop:confluence_Lamr}, the \emph{$\R$-normal form} of $\Set{E}\in\sums{\Lam[r]}$ always exists and is denoted by $\nf[\R]{\Set E}$, i.e.\ $\Set{E}\msto[\R]\nf[\R]{\Set E}\in\sums{\Lam[r]}$ and there is no $\Set E'$ such that $\nf[\R]{\Set E}\to_\R \Set E'$.
Simple terms in $\R$-nf are called ``resource approximants'' because their role is similar to the one played by finite approximants of B\"ohm trees, except that they approximate the normal form of the Taylor expansion.  They admit the following syntactic characterization.

\begin{defi}\label{def:normalized_resterms}
A \emph{resource approximant} $a\in\Lam[s]$ is a simple term generated by the following grammar (for $k,n\ge 0$, where $\bag{x^n}$ is the bag $\bag{x,\dots,x}$ having $n$ occurrences of $x$):
\[
	\begin{array}{lcl}
	a &\bnfeq& b\mid c\\
	b&\bnfeq&\bag{x^n}\mid \bag{\lam x.a_1,\dots,\lam x.a_n}\mid \bag{x}ba_1\cdots a_k\\	
	c&\bnfeq& \bag{\lam x.a}(\bag{y}ba_1\cdots a_k)\\
	\end{array}
\]
\end{defi}
It is easy to check that resource approximants are $\R$-normal forms.

\begin{exa} The following are examples of resource approximants:
\begin{enumerate}
\item\label{ex:resapp1} $\bag{\lam x.\bag{x},\lam x.\bag{x,x},\lam x.\bag{x,x,x}}$ and 
 $\bag{\lam x.\bag{x}\bag{x,x},\lam x.\bag{x}\bag{x,x,x}}$ belong to the Taylor expansion of some \lam-term (as we will see in Example~\ref{exa:TeXeMpLe}).
\item\label{ex:resapp3} $\bag{\lam x.\bag{x,x,x},\lam x.\bag{y,y,y}}$ does not, as will be shown in Proposition~\ref{prop:maximalCl_iff_TeM}.
\end{enumerate}
\end{exa}

\subsection{Characterizing the Taylor Expansion of a \lam-Term}\label{ssec:Te}
We recall the definition of the Taylor expansion of a \lam-term in the \CbV\ setting, following~\cite{Ehrhard12,CarraroG14}.
Such a Taylor expansion translates a \lam-term $M$ into an infinite set\footnote{%
This set can be thought of as the \emph{support} of the actual Taylor expansion, which is an infinite formal linear combination of simple terms taking coefficients in the semiring of non-negative rational numbers.
} of simple terms.
Subsequently, we characterize those sets of resource terms arising as a Taylor expansion of some $M\in\Lam$.

\begin{defi}\label{def:Taylor} 
The \emph{Taylor expansion} $\Te{M}\subseteq\Lam[s]$ of a \lam-term $M$ is an infinite set of simple terms defined by induction as follows:
\[
	\begin{array}{lcl}
	\Te{x}&=&\set{\bag{x^n} \st n\ge 0},\textrm{ where }[x^n] = [x,\dots,x]\textrm{ ($n$ times)},\\
	\Te{\lam x.N}&=&\set{\bag{\lam x.t_1,\dots,\lam x.t_n} \st n\ge 0, \forall i\le n,\ t_i\in\Te{N} },\\
	\Te{PQ}&=&\set{st\st s\in\Te{P},\ t\in\Te{Q}}.\\
	\end{array}
\]
\end{defi}

From the definition above, we get the following easy properties.
\begin{rem}\label{rem:emptymset}\ 
\begin{enumerate}
\item\label{rem:emptymset1} 
	$[]\in\Te{V}$ if and only if $V\in\Val$.
\item\label{rem:emptymset2} 
	Every occurrence of a $\beta_r\sigma$-redex in $t\in\Te{M}$ arises from some $\V$-redex in $M$.
\item\label{rem:emptymset3} 
	By exploiting Notation~\ref{notation:sumsofresterm}, we can rewrite the Taylor expansion of an application or an abstraction as follows:
\[
	\begin{array}{lcl}
	\Te{PQ} &=& \Te{P}\Te{Q},\\
	\Te{\lam x.N} &=& \bigcup_{n\in\nat}\set{\bag{\underbrace{\lam x.\Te{N},\dots,\lam x.\Te{N}}_{n \textrm{ times}}}}.\\
	\end{array}
\] 
\end{enumerate}
\end{rem}

\begin{exa}\label{exa:TeXeMpLe} We calculate the Taylor expansion of some \lam-terms.
\begin{enumerate}
\item $\Te{\comb{I}} = \set{\bag{\lam x.\bag{x^{n_1}},\dots,\lam x.\bag{x^{n_k}}}\st k\ge 0, \forall i\le k, n_i\ge 0 }$,
\item $\Te{\Delta} = \set{ \bag{\lam x.\bag{x^{n_1}}\bag{x^{m_1}},\dots,\lam x.\bag{x^{n_k}}\bag{x^{m_k}}} \st k\ge0,\forall i \le k, m_i,n_i\ge 0}$,
\item $\Te{\Delta\comb{I}} = \set {st\st s\in\Te{\Delta},t\in\Te{\comb{I}}}$,
\item $\Te{\Omega} = \set {st\st s,t\in\Te{\Delta}}$,
\item $\Te{\lambda z.yyz} = \set{\bag{\lambda z.\bag{y^{\ell_1}}\bag{y^{m_1}}\bag{z^{n_1}},\dots,\lambda z.\bag{y^{\ell_k}}\bag{y^{m_k}}\bag{z^{n_k}}} \st k\ge 0,\forall i\le k,\ell_i,m_i,n_i\ge 0}$,
\item $\Te{\lambda y. f(\lambda z.yyz)} = 
\set{
	\bag{\lambda y.\bag{f^{n_1}}t_1,\dots,
	\lambda y.\bag{f^{n_k}}t_k
	} \st k\ge 0,\forall i\le k,n_i\ge 0,t_i\in\Te{\lambda z.yyz}
}$,
\item $\Te{\comb{Z}} = \set {\bag{\lam f.s_1t_1,\dots,\lam f.s_kt_k}\st k\ge 0,\forall i \le k, s_i,t_i\in\Te{\lambda y. f(\lambda z.yyz)}}$.
\end{enumerate}
\end{exa}

These examples naturally brings to formulate the next remark and lemma.

\begin{rem}\label{rem:about0redexes}
An element $t$ belonging to the Taylor expansion of a \lam-term $M$ in $\V$-nf might not be in $\R$-nf, due to the possible presence of $0$-redexes.
For an example, consider $\bag{\lam x.\bag{x,x}\bag{x,x},\lam x.\bag{x}\bag{x,x,x}}\in\Te{\Delta}$.
Notice that, since the reduction does not modify the cardinality of a bag, a more refined definition of Taylor expansion eliminating all 0-redexes is possible by substituting the application case with the following:
\[
	\Te{VM_0\cdots M_k} = \set{\bag{v}t_0\cdots t_k \st \bag{v}\in\Te{V}, \forall i\,.\, (0\le i\le k)\ t_i \in\Te{M_i}}
\]
We prefer  to keep Ehrhard's original notion because it has a simpler inductive definition.
\end{rem}

The following statement concerning the Taylor expansion of \lam-terms in $\V$-nf does hold.

\begin{lem}\label{lem:TeMnfisnf}
For $M\in\Lam$, the following are equivalent:
\begin{enumerate}
\item\label{lem:TeMnfisnf1} $M$ is in $\V$-normal form,
\item\label{lem:TeMnfisnf2} every $t\in\Te{M}$ is in $\beta_r\sigma$-normal form.
\end{enumerate}
\end{lem}

\begin{proof} (\ref{lem:TeMnfisnf1} $\Rightarrow$ \ref{lem:TeMnfisnf2})
Using Lemma~\ref{lem:char_V_normal form}, we proceed by induction on the normal structure of $M$.

If $M = x$ then $t\in\Te{M}$ entails $t = [x,\dots,x]$ which is in $\V$-nf.

If $M = \lam x.G$ then $t\in\Te{M}$ implies that $t = [\lam x.t_1,\dots,\lam x.t_n]$ where $t_i\in\Te{G}$ for all $i\le n$. 
By the induction hypothesis each $t_i$ is in $\beta_r\sigma$-nf, hence, so is $t$.

If $M = xHG_1\cdots G_k$ then $t\in\Te{M}$ entails $t = \bag {x^n}st_1\cdots t_k$ for some $n\ge 0$, $s\in\Te{H}$ and $t_i\in\Te{G_i}$ ($1\le i \le k$). By induction hypothesis $s,t_1,\dots ,t_k$ are in $\beta_r\sigma$-nf, so $t$ is in $\beta_r$-nf.
Concerning $\sigma$-rules, $t$ could have a $\sigma_3$-redex in case $s = [\lam x.s']t'$ but this is impossible since $s\in\Te{H}$ and $H$ cannot have shape $(\lam x.P)Q$.

If $M = (\lam x.G)(yHG_1\cdots G_k)$ and $t\in\Te M$ then $t = [\lam x.s_1,\dots,\lam x.s_n]t'$ for some $n\ge 0$, $s_i\in\Te{G}$, $1\le i \le n$, and $t'\in\Te{yHG_1\cdots G_k}$.
By induction hypothesis, the resource terms $s_1,\dots,s_n$ and $t'$ are in $\beta_r\sigma$-nf. 
In principle, when $n=1$, the simple term $t$ might have the shape either of a $\beta_r$-redex or of a $\sigma_3$-redex. Both cases are impossible since $t'\in\Te{yHG_1\cdots G_k}$ entails $t' = [y^m]st_1\cdots t_k$ which is neither a resource value nor a simple term of shape $[\lam z.s]s'$.
We conclude that $t$ is in $\beta_r\sigma$-nf.

 (\ref{lem:TeMnfisnf2} $\Rightarrow$ \ref{lem:TeMnfisnf1}) We prove the contrapositive. 
 Assume that $M$ is not in $\V$-nf, then either $M$ itself is a $\beta_\V$- or $\sigma$-redex, or it contains one as a subterm. Let us analyze first the former case.
\begin{enumerate}
\item[$(\beta_\V)$] 
 If $M = (\lam x.N)V$ for $V\in\Val$ then, by Remark~\ref{rem:emptymset}\eqref{rem:emptymset1}, the $\beta_r$-redex $[\lam x.s][]$ belongs to $\Te{M}$ for every $s\in\Te{N}$.
\item[$(\sigma_1)$] 
  If $M = (\lam x.N)PQ$ then for all $s\in\Te N, t_1\in\Te P, t_2\in\Te Q$ we have $[\lam x.s]t_1t_2\in\Te{M}$ and this simple term is a $\sigma_1$-redex.
\item[$(\sigma_3)$] 
  If $M = V((\lam x.P)Q)$ for $V\in\Val$ then for all $[v]\in\Te{V},s\in\Te{P}$ and $t'\in\Te{Q}$ we have $[v]([\lam x.s]t')\in\Te{M}$ and this resource term is a $\sigma_2$-redex.
\end{enumerate}  
  Otherwise $M =C\hole{M'}$ where $C$ is a context and $M'$ is a $\V$-redex having one of the shapes above; in this case there is $t\in\Te{M}$ containing a $\beta_r\sigma$-redex $t'\in\Te{M'}$ as a subterm.
\end{proof}

The rest of the section is devoted to provide a characterization of all sets of simple terms that arise as the Taylor expansion of some \lam-term $M$. 
 
\begin{defi}\label{def:heightandstuff}
\begin{enumerate}
\item\label{def:heightandstuff1} The \emph{height of a non-empty set $\Set E\subseteq{\Lam[r]}$}, written $\hgt{\Set E}$, is the maximal height of its elements, if it exists, and in this case we say that $\Set E$ has \emph{finite} height.
Otherwise, we define $\hgt{\Set E} =\aleph_0$ and we say that $\Set E$ has \emph{infinite} height. 
\item\label{def:heightandstuff2} Define a \emph{coherence relation} $\coh\ \subseteq\Lam[r]\times\Lam[r]$ as the smallest relation satisfying:
\[
	\infer{x\coh x}{\phantom{x\coh x}}
	\qquad
	\infer{\lam x.s\coh \lam x.t}{s\coh t}
	\qquad
	\infer{\bag{v_1,\dots,v_k}\coh \bag{v_{k+1},\dots,v_n}}{v_i\coh v_j &(\forall i,j\le n)}
	\qquad
	\infer{s_1t_1\coh s_2t_2}{s_1\coh s_2&t_1\coh t_2}
\]
\item A subset $\Set{E}\subseteq{\Lam[r]}$ is a \emph{clique} whenever $e\coh e'$ holds for all $e,e'\in\Set{E}$. 
\item A clique $\Set E$ is \emph{maximal} if, for every $e\in\Lam[r]$, $\Set E \cup \{e\}$ is a clique entails $e\in\Set E$.
\end{enumerate}
\end{defi}

The coherence relation above is inspired by Ehrhard's work in the call-by-name setting~\cite{EhrhardR08}.
Note that $\coh$ is symmetric, but neither reflexive as $\bag{x,y}\not\coh\bag{x,y}$ nor transitive since $[x] \coh [] \coh [y]$ but $[x] \not\coh [y]$. 

\begin{exa} Notice that all sets in Example~\ref{exa:TeXeMpLe} are maximal cliques of finite height. 
For instance, $\hgt{\Te{\comb{I}}} = 3$ and by following the rules in Definition~\ref{def:heightandstuff}\eqref{def:heightandstuff2} we have $u\coh t$ for all $t\in\Te{\comb{I}}$ if and only if either $u=\bag{}$ or $u = \bag{x^n}$ for some $n\in\nat$ if and only if $u\in\Te{\comb{I}}$. 
Therefore $\Te{\comb{I}}$ is maximal.

The rest of the section is devoted to proving that these two properties actually characterize those sets that the Taylor expansions of \lam-terms (Proposition~\ref{prop:maximalCl_iff_TeM}).
\end{exa}

We may now characterize resource approximants (Definition~\ref{def:normalized_resterms}).

\begin{lem}\label{lem:resourceapprox} 
Let $t\in\Lam[s]$ be such that $t\coh t$. 
Then $t$ is in $\R$-nf iff $t$ is a resource approximant.
\end{lem}

\begin{proof} Notice that $t\coh t$ guarantees that all terms in each bag occurring in $t$ have similar shape.
The proof of the absence of $\beta_r$- and $\sigma$- redexes, is analogous to the one of Lemma~\ref{lem:char_V_normal form}. 
The bags occurring in $\bag{x}ba_1\cdots a_k$ and $\bag{\lam x.a}(\bag{y}ba_1\cdots a_k)$ must be singleton multisets, for otherwise we would have some $0$-redexes.
\end{proof}

This lemma follows easily from Definition~\ref{def:heightandstuff}\eqref{def:heightandstuff1} and  Remark~\ref{rem:emptymset}\eqref{rem:emptymset3}.

\begin{lem}\label{lem:about_hgt} For $N,P,Q\in\Lambda$, we have:
\begin{enumerate}
\item\label{lem:about_hgt1} $\hgt{\Te{\lam x.N}} = \hgt{\Te{N}}+2$.
\item\label{lem:about_hgt2} $\hgt{\Te{PQ}} = \hgt{\Te{P}\cup\Te{Q}}+1$,
\end{enumerate}
\end{lem}
\begin{proof}
\eqref{lem:about_hgt1} Indeed, we have:
\[
	\begin{array}{lcl}
	\hgt{\Te{\lam x.N}}&=&\max\set{\hgt{\bag{\lam x.t_1,\dots,\lam x.t_n}} \st n\ge 0, \forall i\le n,\ t_i\in\Te{N} },\\
	&=&\max\set{\max{\set{\hgt{\lam x.t_1},\dots,\hgt{\lam x.t_n}}}+1 \st n\ge 0, \forall i\le n,\ t_i\in\Te{N} },\\
	&=&\max\set{\max{\set{\hgt{t_1},\dots,\hgt{t_n}}}+2 \st n\ge 0, \forall i\le n,\ t_i\in\Te{N} },\\	
	&=&\max\set{\hgt{t}+2 \st t\in\Te{N} } = \hgt{\Te{N}} + 2.\\		
	\end{array}
\]
\eqref{lem:about_hgt2} This case is analogous but simpler, and we omit it.
\end{proof}

The next proposition gives a characterization of those sets of simple terms corresponding to the Taylor expansion of some \lam-terms and constitutes the main result of the section.

\begin{prop}\label{prop:maximalCl_iff_TeM} 
For $\Set E\subseteq\Lam[s]$, the following are equivalent:
\begin{enumerate}
\item\label{prop:maximalCl_iff_TeM1} $\Set E$ is a maximal clique having finite height,
\item\label{prop:maximalCl_iff_TeM2} There exists $M\in\Lam$ such that $\Set E = \Te M$.
\end{enumerate}
\end{prop}

\begin{proof}
(\ref{prop:maximalCl_iff_TeM1} $\Rightarrow$ \ref{prop:maximalCl_iff_TeM2})
As $\Set E$ maximal entails $\Set E\neq\emptyset$, we can proceed by induction on $h = \hgt{\Set E}$.

The case $h=0$ is vacuous because no simple term has height $0$.

If $h=1$ then $t\in\Set E$ implies $t =\bag{x_1,\dots,x_n}$ since variables are the only resource terms of height $0$. 
Now, $t\coh t$ holds since $\Set E$ is a clique so the $x_i$'s must be pairwise coherent with each other, but $x_i\coh x_j$ holds if and only if $x_i=x_j$ whence $t = \bag{x_i,\dots,x_i}$ for some index $i$.
From this, and the fact that $\Set E$ is maximal, we conclude $\Set E = \Te{x_i}$.

Assume $h > 1$ and split into cases depending on the form of $t\in\Set E$.
\begin{itemize}
\item Case $t = \bag{\lam x.s_1,\dots,\lam x.s_k}$.
Since $\hgt{\Set E} > 1$ we can assume wlog that $t\neq\bag{}$, namely $k>0$.
Moreover, since $\Set E$ is a clique, all $t'\in\Set E$ must have shape $t'=\bag{\lam x.s_{k+1},\dots,\lam x.s_n}$ for some $n$ with $s_i\coh s_j$ for all $i,j\le n$.
It follows that the set $\Set S = \set{s \st \bag{\lam x.s}\in\Set E}$ is a maximal clique, because $\Set E$ is maximal, and has height $h-2$ since $\hgt{\bag{\lam x.s}} = \hgt{s} + 2$. 
Moreover, $\Set E = \set{\bag{\lam x.s_1,\dots,\lam x.s_k} \st k\ge 0, \forall i\le k\ .\ s_i\in\Set S }$.
By induction hypothesis there exists $N\in\Lam$ such that $\Set S = \Te{N}$, so we get $\Set E = \Te {\lam x.N}$.

\item Otherwise, if $t = s_1s_2$ then all $t'\in\Set E$ must be of the form $t'=s'_1s'_2$ with $s_1\coh s'_1$ and $s_2\coh s'_2$. 
So, the set $\Set E$ can be written as $\Set E = \Set{S}_1\Set{S}_2$ where 
$\Set{S}_1= \set{ t\st ts_2\in\Set E}$ and $\Set{S}_2= \set{ t\st s_1t\in\Set E}$.
As $\Set E$ is a maximal clique, the sets $\Set S_1,\Set S_2$ are independent from the choice of $s_2,s_1$ (resp.),  and they are maximal cliques themselves. Moreover, $\hgt{\Set E} = \hgt{\Set S_1 \cup \Set S_2} + 1$, whence the heights of $\Set S_1,\Set S_2$  are strictly smaller than $h$.
By the induction hypothesis, there exists $P,Q\in\Lam$ such that $\Set S_1 = \Te{P}$ and $\Set S_2 = \Te{Q}$, from which it follows $\Set E = \Te{PQ}$.
\end{itemize}

(\ref{prop:maximalCl_iff_TeM2} $\Rightarrow$ \ref{prop:maximalCl_iff_TeM1}) 
We proceed by induction on the structure of $M$.

If $M=x$ then $t,t'\in\Te{M}$ entails $t = [x^k]$ and $t'= \bag{x^n}$ for some $k,n\ge 0$, whence $\Te{x} $ is a clique of height 1.
It is moreover maximal because it contains $\bag{x^i}$ for all $i\ge 0$.

If $M= \lam x.N$ then $t,t'\in\Te{M}$ entails $t = [\lam x.t_1,\dots,\lam x.t_k]$ and $t' = [\lam x.t_{k+1},\dots,\lam x.t_n]$ with $t_i\in\Te{N}$ for all $i\le n$. 
By induction hypothesis $\Te{N}$ is a maximal clique of finite height $h\in\nat$, in particular $t_i\coh t_j$ for all $i,j\le n$ which entails $t\coh t'$. 
The maximality of $\Te{M}$ follows from that of $\Te{N}$ and, by Lemma~\ref{lem:about_hgt}\eqref{lem:about_hgt1}, $\hgt{\Te{M}}$ has finite height $h+2$.

If $M= PQ$ then $t,t'\in\Te{M}$ entails $t = s_1t_1$ and $t'= s_2t_2$ for $s_1,s_2\in\Te{P}$ and $t_1,t_2\in\Te{Q}$. By induction hypothesis, $s_1\coh s_2$ and $t_1\coh t_2$ hold and thus $t\coh t'$.
Also in this case, the maximality of $\Te{M}$ follows from the same property of $\Te{P},\Te{Q}$.
Finally, by induction hypothesis, $\hgt{\Te{P}} = h_1$ and $\hgt{\Te{Q}}= h_2$ for $h_1,h_2\in\nat$ then $\hgt{\Te{M}} = \max\set{h_1,h_2} +1$ by Lemma~\ref{lem:about_hgt}\eqref{lem:about_hgt2}, and this concludes the proof.
\end{proof}

\section{Computing the Normal Form of the Taylor expansion, and Beyond}

The Taylor expansion, as defined in Section~\ref{ssec:Te}, is a static operation translating a \lam-term into an infinite set of simple terms. 
However, we have seen in Proposition~\ref{prop:confluence_Lamr} that the reduction $\to_\R$ is confluent and strongly normalizing.
Whence, it is possible to define the normal form of an arbitrary set of resource terms as follows.

\begin{defi}
The \emph{$\R$-normal form} is extended element-wise to any subset $\Set E\subseteq{\Lam[r]}$ by setting $\NF{\Set E} = \bigcup_{e\in \Set E}\ \nf[\R]{e}$. 
\end{defi}
In particular, $\NF{\Lam[s]}$ (resp.\ $\NF{\Lam[r]}$, $\NF{\Val[r]}$) represents the set of all simple terms (resp.\ resource terms, resource values) in $\R$-nf generated by the grammar in Definition~\ref{def:normalized_resterms}.
Moreover, $\NF{\Te{M}}$ is a well-defined subset of $\NF{\Lam[s]}$ for every $M\in\Lam$ (it can possibly be the empty set, thought).

\begin{exa}\label{ex:quatpointdeux} 
We calculate the $\R$-normal form of the Taylor expansions from Example~\ref{exa:TeXeMpLe}:
\begin{enumerate}
\item\label{ex:quatpointdeux1} $\NF{\Te{\comb{I}}} = \Te{\comb{I}} = \set{\bag{\lam x.\bag{x^{n_1}},\dots,\lam x.\bag{x^{n_k}}}\st k\ge 0, \forall i\le k, n_i\ge 0 }$,
\item\label{ex:quatpointdeux2} $\NF{\Te{\Delta}} = \set{ \bag{\lam x.\bag{x}\bag{x^{m_1}},\dots,\lam x.\bag{x}\bag{x^{m_k}}} \st k\ge0,\forall i \le k, m_i\ge 0}$, 
\item\label{ex:quatpointdeux3} $\NF{\Te{\Delta\comb{I}}} = \NF{\Te{\comb{I}}}$,
\item $\NF{\Te{\Omega}} = \emptyset$, from this it follows:
\item $\NF{\Te{\lam x.\Omega}} = \set{ [] }$, moreover, for $\comb{A} = (\lam z.(\lam y.y)(zz))(xx)$, we obtain:
\item $\NF{\Te{\comb{A}}} = \{\bag{\lam z.\bag{\bag{\lam y.\bag{y^{\ell_1}}}(\bag{z}\bag{z^{m_1}})}{(\bag{x}\bag{x^{n_1}}) },\dots,\lam z.\bag{\bag{\lam y.\bag{y^{\ell_k}}}(\bag{z}\bag{z^{m_k}})}{(\bag{x}\bag{x^{n_k}}) }} \st{}\hspace{75pt}k\ge0,\forall i\le k, \ell_i,m_i,n_i\ge 0\}$.
\end{enumerate}
\end{exa}

On the one hand, it is not difficult to calculate the normal forms of the Taylor expansions of $\comb{I},\Delta$ and $\comb{A}$. (As shown in Lemma~\ref{lem:TeMnfisnf}, it is enough to perform some $0$-reductions.) Similarly, it is not difficult to check that $\NF{\Te{\Omega}}$ is empty, once realized that no term $t\in\Te{\Omega}$ can survive through the reduction.
On the other hand, it is more complicated to compute the normal forms of $\Te{\comb{Z}}$, and hence $\Te{\comb{ZB}}$, without having a result connecting such normal forms with the $\V$-reductions of the corresponding \lam-terms.
The rest of the section is devoted to study such a relationship.
We start with some technical lemmas.

\begin{lem}[Substitution Lemma]\label{lem:TeSubstLemma} 
Let $M\in\Lam$, $V\in\Val$ and $x\in\Var$.
Then we have:
\[
	\Te{M\subst{x}{V}} = \bigcup_{t\in\Te{M}}\bigcup_{[v_1,\dots,v_n]\in\Te{V}
	} t\lsubst{x}{v_1,\dots,v_n}.
\]
\end{lem}
\begin{proof} Straightforward induction on the structure of $M$.
%
%
%
%
\end{proof}

\begin{lem}\label{lem:NFTEM_eq_NFTEN} 
Let $M,N\in\Lam$ be such that $M\to_\V N$. Then:
\begin{enumerate}
\item\label{lem:NFTEM_eq_NFTEN1}  for all $ t\in\Te{M}$, there exists $\Set T\subseteq\Te{N}$ such that $t\msto[\R] \Set T$,
\item\label{lem:NFTEM_eq_NFTEN2} for all $t'\in\Te{N}$ such that $t'\not\to_0 \emptyset$, there exist $t\in\Te{M}$ and $\Set T\in\sums {\Lam[s]}$ satisfying $t\msto[\R] \set {t'}\cup\Set T$. Moreover such a $t$ is unique.
\end{enumerate}
\end{lem}

\begin{proof} We check that both \eqref{lem:NFTEM_eq_NFTEN1} and \eqref{lem:NFTEM_eq_NFTEN2} hold by induction on a derivation of $M \to_\V N$, splitting into cases depending on the kind of redex is reduced.
\begin{description}
\item[$(\beta_v)$] If $M = (\lam x.Q)V$ and $N = Q\subst{x}V$ then items \eqref{lem:NFTEM_eq_NFTEN1} and \eqref{lem:NFTEM_eq_NFTEN2}  follow by Lemma~\ref{lem:TeSubstLemma}.
\item[$(\sigma_1)$] If $M = (\lam x.M')PQ$ and $N = (\lam x.M'Q)P$ then
\[
	\begin{array}{lcl}
	\Te{M} &=& \set{ \bag{\lam x.t_1,\dots,\lam x.t_n}s_1s_2 \st n\ge 0,t_i\in\Te{M'}, s_1\in\Te{P},s_2\in\Te{Q} },\\
	\Te{N} &=&\set{\bag{\lam x.t'_1s'_1,\dots,\lam x.t'_ns'_n}s \st n\ge0, t'_i\in\Te{M'},s'_i\in\Te{Q},s\in\Te{P}}.\\
	\end{array}
\]
For $n\neq 1$, we have $\bag{\lam x.t_1,\dots,\lam x.t_n}s_1s_2\to_0\emptyset\subseteq \Te{N}$.
For $n =1 $, we get $\bag{\lam x.t_1}s_1s_2\to_{\sigma_1}  \bag{\lam x.t_1s_2}s_1$ for $t_1\in\Te{M'}$, $s_1\in\Te{P}$ and $s_2\in\Te{Q}$, whence $\bag{\lam x.t_1s_2}s_1 \in \Te{N}$ and \eqref{lem:NFTEM_eq_NFTEN1} holds.
Concerning \eqref{lem:NFTEM_eq_NFTEN2}, note that $\bag{\lam x.t'_1s'_1,\dots,\lam x.t'_ns'_n}s\not\to_0\emptyset$ entails $n=1$.
Moreover, $\Te{M}\ni\bag{\lam x.t'_1}ss'_1\to_{\sigma_1}\bag{\lam x.t'_1s'_1}s$ since $t'_1\in\Te{M'},s'_1\in\Te{Q},s\in\Te{P}$.

\item[$(\sigma_3)$] If $M = V((\lam x.P)Q)$ for $V\in\Val$ and $N = (\lam x.VP)Q$ then
\[
	\begin{array}{ll}
	\Te{M} = \set{ \bag{v_1,\dots,v_n}(\bag{\lam x.s_1,\dots,\lam x.s_m}s) \st& n\ge 0,\bag{v_1,\dots,v_n}\in\Te{V},\\
	& i\le m,\ s_i\in\Te{P}, s\in\Te{Q}},\\
	\end{array}
\]
\[
	\begin{array}{ll}
	\Te{N} =\set{\bag{\lam x.\bag{v_{11},\dots,v_{1k_1}}s_1,\dots,\lam x.\bag{v_{n1},\dots,v_{nk_n}}s_n}s \st& n\ge0,\ i \le n,\\ 
	&s_i\in\Te{P},s\in\Te{Q},\\
	&\bag{v_{i1},\dots,v_{ik_i}}\in\Te{V}},\\
	\end{array}
\]
For $m\neq 1$ $n\neq 1$, we have $\bag{v_1,\dots,v_n}(\bag{\lam x.s_1,\dots,\lam x.s_m}s)\to_0 \emptyset\subseteq \Te{N}$.
For $m=n=1$, we get $\bag{v_1}(\bag{\lam x.s_1}s)\to_{\sigma_3}\bag{\lam x.\bag{v_1}s_1}s\in\Te{N}$, so \eqref{lem:NFTEM_eq_NFTEN1} holds.
Similarly, we have that $\bag{\lam x.\bag{v_{11},\dots,v_{1k_1}}s_1,\dots,\lam x.\bag{v_{n1},\dots,v_{nk_n}}s_n}s\to_0\emptyset$ whenever $n\neq 1$ or $k_i\neq 1$. For $n=k_1 =1$, we get 
$\Te{M}\ni\bag{v_{11}}(\bag{\lam x.s_1}s)\to_{\sigma_3}\bag{\lam x.\bag{v_{11}}s_1}s$ which proves \eqref{lem:NFTEM_eq_NFTEN2}.
\end{description}
In the cases above it is easy to check that $t$ is actually unique.
The contextual cases follow straightforwardly from the induction hypothesis.
\end{proof}

As a consequence, we obtain the analogue of Proposition~\ref{prop:M=VNimpBTM=BTN} for Taylor expansions.

\begin{cor}\label{cor:NFTEM_eq_NFTEN} 
For $M,N\in\Lam$, $M =_\V N$ entails $\NF{\Te{M}} = \NF{\Te{N}}$.
\end{cor}
\begin{proof} It is enough to prove $\NF{\Te{M}} = \NF{\Te{N}}$ for $M$ and $N$ such that $M\to_\V N$, indeed the general result follows by confluence of $\V$-reduction. We show the two inclusions.

$(\subseteq)$ Consider $t\in\NF{\Te{M}}$, then there exists $t_0\in\Te{M}$ and $\Set T\in\sums{\Lam[s]}$ such that $t_0\msto[\R] \set t\cup \Set T$. 
Since $\lamr$ is strongly normalizing (Proposition~\ref{prop:confluence_Lamr}), we assume wlog $\Set T$ in $\R$-nf.
By Lemma~\ref{lem:NFTEM_eq_NFTEN}\eqref{lem:NFTEM_eq_NFTEN1}, we have $t_0\msto[\R]\Set T_0\subseteq{\Te{N}}$ so by confluence of $\to_r$ we get $\Set T_0\msto[\R] \set t\cup\Set T$ which entails $t\in\NF{\Te{N}}$ because $t$ is in $\R$-nf.

$(\supseteq)$ If $t\in\NF{\Te{N}}$ then there are $s\in\Te{N}$ and $\Set T\in\sums{\Lam[s]}$ such that $s\msto[\R] \set t\cup \Set T$. 
By Lemma~\ref{lem:NFTEM_eq_NFTEN}\eqref{lem:NFTEM_eq_NFTEN2}, there exists $s_0\in\Te{M}$ and $\Set S\in\sums {\Lam[s]}$ satisfying $s_0\msto[\R] \set s\cup\Set S$.
Composing the two reductions we get $s_0\msto[\R] \set t\cup\Set S\cup\Set T$, thus $t\in\NF{\Te{M}}$ as well.
\end{proof}

We now prove a Context Lemma for Taylor expansions in the spirit of \cite[Cor.~14.3.20]{Bare} (namely, the Context Lemma for \CbN{} B\"ohm trees).
For the sake of simplicity, in the next lemma we consider head contexts but the same reasoning works for arbitrary contexts.

\begin{lem}[Context Lemma for Taylor expansions]\label{lem:Te_ctxt} 
Let $M,N\in\Lam$. If $\NF{\Te{M}} = \NF{\Te{N}}$ then, for all head contexts $C\hole-$, we have $\NF{\Te{C\hole M}} = \NF{\Te{C\hole N}}$.
\end{lem}

\begin{proof} 
Consider $C\hole- = (\lam x_1\dots x_n.\hole-)V_1\cdots V_k$ for $n,k\ge 0$.
Let us take $t\in\NF{\Te{C\hole M}}$ and prove that $t$ belongs to $\NF{\Te{C\hole N}}$, the other inclusion being symmetrical.
Then there exists $t_0\in\Te{C\hole M}$ and $\Set T\in\sums{\NF{\Lam[s]}}$ such that $t_0\msto[\R] \set t \cup \Set T$.
By definition of $C\hole-$ and $\Te-$, $t_0$ must have the following shape:
\[
	t_0 = \bag{\lam x_1.\bag{\cdots\bag{\lam x_n.s}\cdots}}\bag{v_{11},\dots,v_{1n_1}}\cdots\bag{v_{k1},\dots,v_{kn_k}}
\]
where $s\in\Te{M}$, $\bag{v_{i1},\dots,v_{in_i}}\in\Te{V_i}$ where $1\le i\le k$ and $n_i = \deg_{x_i}(s)$ for otherwise $t_0\to_0 \emptyset$, which is impossible.
By confluence and strong normalization of $\to_\R$ (Proposition~\ref{prop:confluence_Lamr}), the reduction $t_0\msto[\R] \set t \cup \Set T$ factorizes as $t_0 \msto[\R] \Set T_0 \msto[\R] \set t \cup \Set T$ where 
\[
	\Set T_0 = \bag{\lam x_1.\bag{\cdots\bag{\lam x_n.\nf[\R]s}\cdots}}\bag{v_{11},\dots,v_{1n_1}}\cdots\bag{v_{k1},\dots,v_{kn_k}}
\]
and $\nf[\R]s \in\sums{\NF{\Te{M}}}$.
By hypothesis $\nf[\R]s \in\sums{\NF{\Te{N}}}$, therefore there are $\Set S_1\in\sums{\Te{N}}$ such that
$\Set S_1 \msto[\R] \nf[\R]s\cup\Set S'$, for some $\Set S'$, and $\Set S_0\subseteq\Te{C\hole N}$ of shape
\[
	\Set S_0 = \bag{\lam x_1.\bag{\cdots\bag{\lam x_n.\Set S_1}\cdots}}\bag{v_{11},\dots,v_{1n_1}}\cdots\bag{v_{k1},\dots,v_{kn_k}}
\]
so we conclude, for some $\Set S''$, that $\Set S_0\msto[\R] \Set T_0 \cup \Set S''\msto[\R] \set t \cup \Set T\cup \nf[\R]{\Set S''}\subseteq\NF{\Te{C\hole N}}$.
\end{proof}

\subsection{Taylor expanding B\"ohm trees.}

The Taylor expansion can be extended to elements of $\Lamb$ by adding $\Te{\bot} = \set{\bag{}}$ to the rules of Definition~\ref{def:Taylor}.
However, the resulting translation of an approximant $A$ produces a set of resource terms that are not necessarily in $\R$-normal form because of the presence of $(0)$-redexes (as already discussed in Remark~\ref{rem:about0redexes}).
Luckily, it is possible to slightly modify such a definition by performing an ``on the flight'' normalization and obtain directly the normalized Taylor expansion of a B\"ohm tree.

\begin{defi}\label{def:TE_BT}
\begin{enumerate}
\item
	Let $A\in\App$. The \emph{normalized Taylor expansion of $A$}, in symbols $\Ten{A}$, is defined by structural induction following the grammar of Definition~\ref{def:approx}\eqref{def:approx1}:
	\[
		\begin{array}{rcl}
		\Ten{x}&=&\set{\bag{x^n}\st n\ge 0},\\[1ex]
		\Ten{\lam x.A'}&=&\set{[\lam x.t_1,\dots,\lam x.t_n] \st n\ge 0, \forall i\le k\ .\ t_i\in\Ten{A'} },\\[1ex]
		\Ten{\bot}&=&\set{\bag{}},\\[1ex]
		\Ten{xBA_1\cdots A_k}&=&\set{[x]t_0\cdots t_n\st t_0\in\Ten{B}, \forall 1\le i\le k\ .\ t_i\in\Ten{A_i}},\\[1ex]
	\Ten{(\lam x.A')(yBA_1\cdots A_k)}&=&\set{\bag{\lam x.s}t\st s\in\Ten{A'},\ t\in\Ten{yBA_1\cdots A_k}}.\\
		\end{array}
	\]
\item
	The \emph{normalized Taylor expansion of $\BT{M}$}, written $\Ten{\BT{M}}$, is defined by setting:
\[
	\Ten{\BT{M}} = \bigcup_{A\in\AM} \Ten{A}
\]
\end{enumerate}
\end{defi}

\begin{exa}
\begin{enumerate}
\item Recall from Example~\ref{ex:deuxpointdeux}\eqref{ex:deuxpointdeux1} that $\Appof{\comb{I}} = \set{ \bot,\lam x.\bot,\lam x.x}$, therefore 
\[
	\Ten{\Appof{\comb{I}}} = 
	\set{\bag{}}\cup
	\set{\bag{(\lam x.\bag{})^k}\st k\ge 0}\cup
	\set{\bag{\lam x.\bag{x^{n_1}},\dots,\lam x.\bag{x^{n_k}}} \st 	k,n_1,\dots,n_k\ge 0}
\]
By Example~\ref{ex:quatpointdeux}\eqref{ex:quatpointdeux1} this is equal to $\NF{\Te{\comb{I}}}$.
\item Since $\Appof{\Omega} = \emptyset$ we have $\Ten{\Appof{\Omega}} = \emptyset = \NF{\Te{\Omega}}$.
\item Also, $\Appof{\Delta} = \{\bot, \lambda x.\bot, \lambda x.xx\}$, so that
\[
	\Ten{\Appof{\Delta}} = 
	\set{\bag{}}\cup
	\set{\bag{(\lam x.\bag{})^k}\st k\ge 0}\cup
	\set{\bag{\lam x.\bag{x}\bag{x^{n_1}},\dots,\lam x.\bag{x}\bag{x^{n_k}}} \st 	k,n_1,\dots,n_k\ge 0}
\]
By Example~\ref{ex:quatpointdeux}\eqref{ex:quatpointdeux2} this is equal to $\NF{\Te{\Delta}}$.
\item Finally, Examples~\ref{ex:deuxpointdeux}\eqref{ex:deuxpointdeux3}  and~\ref{ex:quatpointdeux}\eqref{ex:quatpointdeux3} and the above item (1) give us  $\Ten{\Appof{\Delta\comb I}} = \Ten{\Appof{\comb I}}= \NF{\Te{\comb I}}= \NF{\Te{\Delta\comb I}}$. 
\end{enumerate}
\end{exa}

The rest of the section is devoted to generalizing the above example, proving that the normal form of the Taylor expansion of any \lam-term $M$ is equal to the normalized Taylor expansion of the B\"ohm tree of $M$ (Theorem~\ref{thm:T_commutes_with_BT}).
On the one side, this link is extremely useful to compute $\NF{\Te{M}}$ because the B\"ohm trees have the advantage of hiding the explicit amounts of resources that can become verbose and difficult to handle.
On the other side, this allow to transfer results from the Taylor expansions to B\"ohm trees, Lemma~\ref{lem:ctxt} being a paradigmatic example.

\begin{lem}\label{lem:TE_tech} Let $M\in\Lam$.
\begin{enumerate}[beginpenalty=99,midpenalty=99]
\item\label{lem:TE_tech0} 
	If $t\in\Te{M}$ and $t\to_\R \set{t_1}\cup\Set T_1$, then there exists $N\in\Lam$ and $\Set T_2\in\sums{\Lam[s]}$, such that $M\to_\V N$ and $\set{t_1}\cup\Set T_1\msto[\R]\Set T_2\subseteq\Te{N}$.
\item\label{lem:TE_tech1} 
	If $t\in\NF{\Te{M}}$ then there exists $M'$ such that  $M\msto M'$ and $t\in\Te{M'}$.
\item\label{lem:TE_tech1half}  
	If $t,s\in\Te{M}$ then $\NF{t}\cap\NF{s}\neq \emptyset$ entails $t = s$.
\item\label{lem:TE_tech2} 
	If $t\in\Te{M}\cap\NF{\Lam[s]}$ then there exists $A\in\App$ such that $A\BTle M$ and $t\in\Ten{A}$.
\end{enumerate}
\end{lem}

\begin{proof}\eqref{lem:TE_tech0}
Note that $t\to_\R\set{t_1}\cup\Set T_1$ by contracting an $\R$-redex arising from an occurrence of a $\V$-redex in $M$, so $M\to_\V N$ where $N$ is obtained by contracting such a redex occurrence.
By Lemma~\ref{lem:NFTEM_eq_NFTEN}\eqref{lem:NFTEM_eq_NFTEN1} and confluence of $\to_\R$, there exists $\Set T_2\subseteq \Te{N}$ such that $t\msto[\R] \set{t_1}\cup\Set T_1\msto[\R] \Set T_2$.

\eqref{lem:TE_tech1} Assume that $t\in\NF{\Te{M}}$, then there are $t_0\in\Te{M}$ and $\Set T\in\sums{\Lam[s]}$ such that $t_0\msto[\R] \set t\cup\Set T$.
Since $\to_\R$ is strongly normalizing, we can assume $\Set T\subseteq\NF{\Lam[s]}$ and choose such a reduction to have maximal length $n$. 
We proceed by induction on $n$ to show that the \lam-term $M'$ exists.
If $n=0$ then $t_0$ is in $\R$-nf so just take $t_0=t$, $\Set T = \emptyset$ and $M = M'$.
Otherwise $n>0$ and $t_0\to_\R \set{t_1}\cup \Set T_1\msto[\R] \set{t}\cup\Set T$ where the second reduction is strictly shorter. 
By \eqref{lem:TE_tech0} and confluence there exists $N$ such that $M\to_\V N$ and $\set{t_1}\cup\Set T_1\msto[\R]\Set T_2\msto[\R] \set{t}\cup \Set T$ for some $\Set T_2\subseteq\Te{N}$.
So, there are $t_2\in\Set T_2$ and $\Set T'\in\sums{\Lam[s]}$ such that $t_2\msto[\R] \set{t}\cup\Set T'\subseteq\NF{\Te{M}}$ so we conclude by applying the induction hypothesis to this reduction  shorter than $n$.

\eqref{lem:TE_tech1half} Assume $t_0\in\NF{t}\subseteq\NF{\Te{M}}$. By \eqref{lem:TE_tech1}  there is a reduction $M \to_\V M_1 \to_\V\cdots \to_\V M_k$ such that $t_0\in\Te{M_k}$. 
By an iterated application of Lemma~\ref{lem:NFTEM_eq_NFTEN}\eqref{lem:NFTEM_eq_NFTEN2}, we get that $t$ is the unique element in $\Te{M}$ generating $t_0$. Therefore, $t_0\in\NF{s}$ entails $s = t$.

\eqref{lem:TE_tech2} By structural induction on the normal structure of $t$ (characterized in Lemma~\ref{lem:resourceapprox}: notice that $t\coh t$ by Proposition~\ref{prop:maximalCl_iff_TeM}).

If $t = \bag{}$ then $M\in\Val$ and there are two subcases: either $M=x$, or $M=\lam x.M'$ so we simply take $A = \bot$.
Similarly, if $t = \bag{x,\dots, x}$ ($n>0$ occurrences) then $M = A = x.$

If $t = \bag{\lam x.a_1,\dots,\lam x.a_n}$ with $n>0$ then $M = \lam x.M'$ and $a_i\in\Te{M'}$ for $i\le n$. By induction hypothesis, there are approximants $A_i\BTle M'$  such that $a_i\in\Ten{A_i}$.
Then we set $A = \lam x. A'$ for $A' = A_1\sqcup\cdots \sqcup A_n$ which exists because the $A_i$'s are pairwise compatible.

If $t = \bag{x}ba_1\cdots a_k$ then $M = xM_0\cdots M_k$ and $b\in\Te{M_0}$ and $a_j\in\Te{M_i}$ for $1\le j \le k$. 
By induction hypothesis, there are $A_0,\dots,A_k$ such that $A_i\BTle M_i$ for all $i$ ($0\le i\le k$), $b\in\Ten{A_0}$ and 
$a_j\in\Ten{A_j}$ for $1\le j \le k$. Moreover $b\in\Ten{A_0}$ entails that $A_0$ is a $B$-term from the grammar in Lemma~\ref{lem:char_V_normal form}, therefore we may take $A = xA_0\cdots A_k\in\App$.

Finally, if $t = \bag{\lam x.a}(\bag{y}ba_1\cdots a_k)$ then $M = (\lam x.M')(yM_0\cdots M_k)$ with $a\in\Te{M'}$ and $\bag{y}ba_1\cdots a_k\in\Te{yM_0\cdots M_k}$.
Reasoning as in the previous case, we get $yA_0\cdots A_k\in\App$ such that $\bag{y}ba_1\cdots a_k\in\Ten{yA_0\cdots A_k}$. 
Moreover, by induction hypothesis, there is $A'\BTle M'$ such that $a\in\Ten{A'}$.
We conclude by taking $A = (\lam x.A')(yA_0\cdots A_k)$.
\end{proof}

\begin{lem}\label{lem:mix} 
Let $M\in\Lam$ and $A\in\App$.
\begin{enumerate}
\item\label{lem:MleA_imp_TenAsseqNFTEM}
If $A\BTle M$ then $\Ten{A}\subseteq \NF{\Te{M}}$.
\item\label{lem:ManRuop14} 
If $\Ten{A}\subseteq\Ten{\BT{M}}$ then $A\in\Appof{M}$.
\end{enumerate}
\end{lem}

\begin{proof} \eqref{lem:MleA_imp_TenAsseqNFTEM} If $A= \bot$ then $M\in\Val$ and $\Ten{\bot} = \set{[]} \subseteq\Te{M} \cap \NF{\Lam[r]}$. 

Otherwise, it follows by induction on $A$ exploiting the fact that all simple terms in $\Ten{A}$ belong to $\Te{M}$ and are already in $\R$-nf.

\eqref{lem:ManRuop14} We proceed by structural induction on $A$, the case $A =\bot$ being trivial.
\begin{itemize}
\item If $A = x$ then $\Ten{x} = \set{[x^n] \st n\ge 0}\subseteq\Ten{\BT M}$ entails $M \msto x$ and we are done.
\item If $A = \lam x.A'$ then $\Ten{\lam x.A'} = \set{[\lam x.t_1,\dots,\lam x.t_n] \st n\ge 0, \forall i\le k\ .\ t_i\in\Ten{A'} }$.
So, $\Ten{\lam x.A'}\subseteq\Ten{\BT{M}}$ implies that $M\msto \lam x.M'$ for some $M'$ such that $\Ten{A'}\subseteq\Ten{\BT{M'}}$.
By induction hypothesis, we get $A'\in\Appof {M'}$ and $\lam x.A'\in \Appof {\lam x.M'}$.
By Lemma~\ref{lem:AM_preserved_by_toV} we obtain $\lam x.A'\in \Appof {M}$ as desired.
\item If $A = \bot$, then $\Ten{A} = \set{\bag{}}\subseteq\Ten{\BT M}$ entails $M\msto V$ for some value $V$, therefore we get $\bot\in\Appof {V}$ and we conclude by Lemma~\ref{lem:AM_preserved_by_toV}.
\item If $A = xBA'_1\cdots A'_k$ then $\Ten{A}=\set{[x]t_0\cdots t_n\st t_0\in\Ten{B}, \forall 1\le i\le k\ .\ t_i\in\Ten{A'_i}}$.
In this case, we must have $M\msto xM_0\cdots M_k$ with $\Ten{B}\subseteq\Ten{\BT{M_0}}$ and $\Ten{A'_i}\subseteq\Ten{\BT{M_i}}$ for $1\le i \le k$.
By induction hypothesis $B \in  \mathcal{A}(M_0)$ and $A'_i \in  \mathcal{A}(M_i)$ $\forall i \in  \{1,\dots,k\}$, thus $xBA'_1\cdots A'_k\in\Appof{xM_0\cdots M_k} = \AM$ by Lemma~\ref{lem:AM_preserved_by_toV}.

\item If $A = (\lam x.A')(yBA'_1\cdots A'_k)$, then $\Ten{A}=\set{\bag{\lam x.s}t\st s\in\Ten{A'},\ t\in\Ten{yBA'_1\cdots A'_k}}$.
In this case we get $M\msto (\lam x.M')(yM_0\cdots M_k)$ with $\Ten{A'}\subseteq\Ten{\BT{M'}}$ and $\Ten{yBA'_1\cdots A'_k}\subseteq\Ten{\BT{yM_0\cdots M_k}}$. 
By applying the induction hypothesis, we obtain $A\in\Appof{(\lam x.M')(yM_0\cdots M_k)}$ and once again we conclude by Lemma~\ref{lem:AM_preserved_by_toV}.\qedhere
\end{itemize}
\end{proof}

The following constitutes the main result of the section, relating B\"ohm trees and Taylor expansion in the spirit of~\cite{EhrhardR06}.

\begin{thm}\label{thm:T_commutes_with_BT} 
For all $M\in\Lam$, we have $\Ten{\BT{M}} = \NF{\Te{M}}$.
\end{thm}

\begin{proof} $(\subseteq)$ Take $t\in\Ten{\BT M}$, then there exists an approximant $A'\in\Appof{M}$ such that $t\in\Ten{A'}$.
As $A'\in\Appof{M}$, there is $M'\in\Lam$ such that $M\msto M'$ and $A'\BTle M'$.
We can therefore apply Lemma~\ref{lem:mix}\eqref{lem:MleA_imp_TenAsseqNFTEM} to conclude that $t\in\NF{\Te{M'}}$, which is equal to $\NF{\Te M}$ by Lemma~\ref{cor:NFTEM_eq_NFTEN}.

$(\supseteq)$ Assume $t\in\NF{\Te{M}}$. By Lemma~\ref{lem:TE_tech}\eqref{lem:TE_tech1} there exists $M'\in\Lam$ such that $M\msto M'$ and $t\in\Te{M'}$.
By Lemma~\ref{lem:TE_tech}\eqref{lem:TE_tech2}, there is $A\BTle M'$ such that $t\in\Ten{A}$.
By the conditions above we have $A\in\Appof{M}$, so we conclude that $t\in\Ten{\BT{M}}$.
\end{proof}

\subsection{Consequences of the main theorem}\label{subsect:cons}

The rest of the section is devoted to present some interesting consequences of Theorem~\ref{thm:T_commutes_with_BT}.

\begin{cor}\label{cor:BTs_iff_Ten}
For $M,N\in\Lambda$, the following are equivalent:
\begin{enumerate}
\item $\BT{M} = \BT{N}$,
\item $\NF{\Te{M}} = \NF{\Te{N}}$.
\end{enumerate}
\end{cor}

\begin{proof}
$(1\Rightarrow 2)$ If $M,N$ have the same B\"ohm tree, we can apply Theorem~\ref{thm:T_commutes_with_BT} to get
\[
	\NF{\Te{M}} = \Ten{\BT{M}}= \Ten{\BT{N}}= \NF{\Te{N}}.
\]

$(1\Leftarrow 2)$ We assume $\NF{\Te{M}} = \NF{\Te{N}}$ and start showing $\Appof{M} \subseteq \Appof{N}$.
Take any $A\in\Appof{M}$, by definition we have $\Ten{A}\subseteq\Ten{\BT{M}}$, so Lemma~\ref{lem:mix}\eqref{lem:ManRuop14} entails $A\in\BT{N}$.
The converse inclusion being symmetrical, we get $\AM = \Appof N$ which in its turn entails $\BT{M} = \BT{N}$ by Remark~\ref{rem:Appof_iff_BT}.
\end{proof}


Carraro and Guerrieri showed in~\cite{CarraroG14} that the relational model $\mathscr{U}$ of \CbV{} \lam-calculus and resource calculus introduced by Ehrhard in~\cite{Ehrhard12} satisfies the $\sigma$-rules, so it is actually a model of both $\lsv$ and $\lambda_r^\sigma$.
They also prove that $\mathscr{U}$ satisfies the Taylor expansion in the following technical sense (where $\Int-$ represents the interpretation function in $\mathscr{U}$):
\begin{equation}\label{eq:CarGue}
	\Int{M} = \bigcup_{t\in\Te{M}}\Int{t}
\end{equation}
As a consequence, we get that the theory of the model $\mathscr{U}$ is included in the theory equating all \lam-terms having the same B\"ohm trees.
We conjecture that the two theories coincide.

\begin{thm} For $M,N\in\Lam$, we have:
\[
	\BT{M} = \BT{N}\ \Rightarrow\ \Int{M}= \Int{N}.
\]
\end{thm}
\begin{proof}Indeed, we have the following chain of equalities:
\[
	\begin{array}{rll}
	\Int{M} =& \bigcup_{t\in\Te{M}}\Int{t},&\textrm{by \eqref{eq:CarGue},}\\
	=& \bigcup_{t\in\NF{\Te{M}}}\Int{t},&\textrm{as }\Int{t} = \bigcup_{s\in\nf[\R]{t}}\Int{s} ,\\
	=& \bigcup_{t\in\Ten{\BT{M}}}\Int{t},&\textrm{by Theorem~\ref{thm:T_commutes_with_BT}},\\
	=& \bigcup_{t\in\Ten{\BT{N}}}\Int{t},&\textrm{as }\BT{M}=\BT{N},\\	
	=& \bigcup_{t\in\NF{\Te{N}}}\Int{t},&\textrm{by Theorem~\ref{thm:T_commutes_with_BT}},\\
	=& \bigcup_{t\in{\Te{N}}}\Int{t},&\textrm{as }\Int{t} = \bigcup_{s\in\nf[\R]{t}}\Int{s} ,\\	
	=&\Int{N},&\textrm{by \eqref{eq:CarGue}.}\\
	\end{array}
\]
This concludes the proof.
\end{proof}

In the paper~\cite{CarraroG14}, the authors also prove that $\Int{M}\neq\emptyset$ exactly when $M$ is potentially valuable (Definition~\ref{def:pot_val}). From this result, we obtain easily the lemma below.

\begin{thm}\label{thm:MPViffBTMnotbot} 
For $M\in\Lam$, the following are equivalent:
\begin{enumerate}
\item $M$ is potentially valuable, 
\item $\BT{M}\neq\bot$.
\end{enumerate}
\end{thm}

\begin{proof} It is easy to check that all resource approximants $t$ have non-empty interpretation in $\mathscr{U}$, i.e.\ $t\in\NF{\Lam[s]}$ entails $\Int{t}\neq\emptyset$. 
Therefore we have the following chain of equivalences:
\[
	\begin{array}{rll}
	M\textrm{ potentially valuable} \iff &\Int{M}\neq\emptyset,&\textrm{by \cite[Thm.~24]{CarraroG14},}\\
	\iff&\exists t\in\NF{\Te{M}}, \Int{t}\neq\emptyset&\textrm{by \eqref{eq:CarGue},}\\
	\iff&\exists s\in\Ten{\BT{M}}, \Int{s}\neq\emptyset&\textrm{by Theorem~\ref{thm:T_commutes_with_BT}},\\
	\iff&\exists A\in\Appof{M},A\neq\bot\\
	\end{array}
\]
This is equivalent to say that $\BT{M}\neq\bot$.
\end{proof}

After this short, but fruitful, semantical digression we conclude proving that all \lam-terms having the same B\"ohm tree are indistinguishable from an observational point of view.
As in the \CbN{} setting, also in \CbV{} this result follows from the Context Lemma for B\"ohm trees.
The classical proof of this lemma in \CbN{} is obtained by developing an interesting, but complicated, theory of syntactic continuity (see \cite[\S14.3]{Bare} and \cite[\S2.4]{AmadioC98}).
Here we bypass this problem completely, and obtain such a result as a corollary of the Context Lemma for Taylor expansions by applying Theorem~\ref{thm:T_commutes_with_BT}.

\begin{lem}[Context Lemma for B\"ohm trees]\label{lem:ctxt} Let $M,N\in\Lam$. If $\BT{M} = \BT{N}$ then, for all head contexts $C\hole-$, we have $\BT{C\hole M} = \BT{C\hole N}$.
\end{lem}

\begin{proof} It follows from the Context Lemma for Taylor expansions (Lemma~\ref{lem:Te_ctxt}) by applying Theorem~\ref{thm:T_commutes_with_BT} and Corollary~\ref{cor:BTs_iff_Ten}.
\end{proof}

As mentioned in the discussion before Lemma~\ref{lem:Te_ctxt}, both the statement and the proof generalize to arbitrary contexts. Thanks to Remark~\ref{rem:contextlemmaholds}, we only need head contexts in order to prove the following theorem stating that the B\"ohm tree model defined in this paper is adequate for Plotkin's \CbV{} \lam-calculus.

\begin{thm}\label{thm:happy_ending} 
Let $M,N\in\Lam$. If $\BT{M} = \BT{N}$ then $M\obseq N$.
\end{thm}

\begin{proof} Assume, by the way of contradiction, that $\BT{M} = \BT{N}$ but $M\not\obseq N$.
Then, there exists a head context $C\hole-$ such that $C\hole M,C\hole N\in\Lambda^o$ and, say, $C\hole M$ is valuable while $C\hole N$ is not. 
Since they are closed \lam-terms, this is equivalent to say that $C\hole M$ is potentially valuable while $C\hole N$ is not.
By Theorem~\ref{thm:MPViffBTMnotbot}, $\BT{C\hole M} \neq\bot$ and $\BT{C\hole N} = \bot$.
As a consequence, we obtain $\BT{C\hole M}\neq \BT{C\hole N}$ thus contradicting the Context Lemma for B\"ohm trees (Lemma~\ref{lem:ctxt}).
\end{proof}

Notice that the converse implication does not hold --- for instance it is easy to check that $\Delta(yy)\obseq yy(yy)$ holds, but the two \lam-terms have different B\"ohm trees.


\section{Conclusions}

Inspired by the work of Ehrhard~\cite{Ehrhard12}, Carraro and Guerrieri~\cite{CarraroG14}, we proposed a notion of B\"ohm tree for Plotkin's call-by-value \lam-calculus $\lam_v$, having a strong mathematical background rooted in Linear Logic.
We proved that \CbV{} B\"ohm trees provide a syntactic model of $\lam_v$ which is adequate (in the sense expressed by Theorem~\ref{thm:happy_ending}) but not fully abstract --- there are operationally indistinguishable \lam-terms having different B\"ohm trees.
The situation looks similar in call-by-name where one needs to consider B\"ohm trees up to possibly infinite $\eta$-expansions to capture the \lam-theory $\mathcal{H}^*$ and obtain a fully abstract model~\cite[Cor.~19.2.10]{Bare}.
Developing a notion of extensionality for \CbV{} B\"ohm trees is certainly interesting, as it might help to describe the equational theory of some extensional denotational model, and a necessary step towards full abstraction. 
Contrary to what happens in call-by-name, this will not be enough to achieve full abstraction as shown by the counterexample $\Delta(yy)\obseq yy(yy)$ but $\BT{\Delta(yy)}\neq \BT{yy(yy)}$, where extensionality plays no role.
The second and third authors, together with Ronchi Della Rocca, recently introduced in~\cite{ManzonettoPR18} a new class of adequate models of $\lsv$ and showed that they validate not only $=_\V$ but also some $I$-reductions (in the sense of $\lam I$-calculus~\cite[Ch.~9]{Bare}) preserving the operational semantics of \lam-terms.
Finding a precise characterization of those $I$-redexes that can be safely contracted in the construction of a \CbV{} B\"ohm tree is a crucial open problem that can lead to full abstraction.


\section*{Acknowledgment}
  \noindent The authors wish to acknowledge fruitful discussions with Giulio Guerrieri, Luca Paolini and Simona Ronchi della Rocca.

\bibliographystyle{alpha}
\bibliography{include/biblio}

\newcommand{\online}[1]{Available at \url{#1}}
\begin{thebibliography}{RDRP04}

\bibitem[AC98]{AmadioC98}
Roberto Amadio and Pierre-Louis Curien.
\newblock {\em Domains and Lambda Calculi}.
\newblock Cambridge tracts in theoretical computer science. Cambridge
  University Press, 1998.

\bibitem[AG17]{AccattoliG17}
Beniamino Accattoli and Giulio Guerrieri.
\newblock Implementing open call-by-value (extended version).
\newblock {\em CoRR}, abs/1701.08186, 2017.

\bibitem[Bar77]{Bare77}
Henk~P. Barendregt.
\newblock The type free lambda calculus.
\newblock In Jon Barwise, editor, {\em Handbook of Mathematical Logic},
  volume~90 of {\em Studies in Logic and the Foundations of Mathematics}, pages
  1091 -- 1132. Elsevier, 1977.

\bibitem[Bar84]{Bare}
Henk~P. Barendregt.
\newblock {\em The lambda-calculus, its syntax and semantics}.
\newblock Number 103 in Studies in Logic and the Foundations of Mathematics.
  North-Holland, second edition, 1984.

\bibitem[BDS13]{Bare2}
Henk~P. Barendregt, Wil Dekkers, and Richard Statman.
\newblock {\em Lambda Calculus with Types}.
\newblock Perspectives in logic. Cambridge University Press, 2013.

\bibitem[BHP13]{BoudesHP13}
Pierre Boudes, Fanny He, and Michele Pagani.
\newblock A characterization of the {T}aylor expansion of lambda-terms.
\newblock In Simona Ronchi~Della Rocca, editor, {\em Computer Science Logic
  2013 {(CSL} 2013), {CSL} 2013, September 2-5, 2013, Torino, Italy}, volume~23
  of {\em LIPIcs}, pages 101--115. Schloss Dagstuhl - Leibniz-Zentrum fuer
  Informatik, 2013.

\bibitem[Bö68]{boehm68}
Corrado Böhm.
\newblock Alcune propriet\`a delle forme $\beta$-$\eta$-normali nel
  $\lambda$-{$K$}-calcolo.
\newblock {\em INAC}, 696:1--19, 1968.

\bibitem[CG14]{CarraroG14}
Alberto Carraro and Giulio Guerrieri.
\newblock {A Semantical and Operational Account of Call-by-Value Solvability}.
\newblock In Anca Muscholl, editor, {\em {Foundations of Software Science and
  Computation Structures}}, volume 8412 of {\em {Lecture Notes in Computer
  Science}}, pages 103--118. Springer-Verlag, 2014.

\bibitem[Ehr12]{Ehrhard12}
Thomas Ehrhard.
\newblock Collapsing non-idempotent intersection types.
\newblock In P.~C{\'{e}}gielski and A.~Durand, editors, {\em Computer Science
  Logic (CSL'12), 21st Annual Conference of the EACSL, {CSL} 2012}, volume~16
  of {\em LIPIcs}, pages 259--273. Schloss Dagstuhl - Leibniz-Zentrum fuer
  Informatik, 2012.

\bibitem[ER03]{EhrhardR03}
Thomas Ehrhard and Laurent Regnier.
\newblock The differential lambda-calculus.
\newblock {\em Theor. Comput. Sci.}, 309(1-3):1--41, 2003.

\bibitem[ER06]{EhrhardR06}
Thomas Ehrhard and Laurent Regnier.
\newblock B{\"{o}}hm trees, {K}rivine's machine and the {T}aylor expansion of
  lambda-terms.
\newblock In Arnold Beckmann, Ulrich Berger, Benedikt L{\"{o}}we, and John~V.
  Tucker, editors, {\em Logical Approaches to Computational Barriers, Second
  Conference on Computability in Europe, CiE 2006, Swansea, UK, June 30-July 5,
  2006, Proceedings}, volume 3988 of {\em Lecture Notes in Computer Science},
  pages 186--197. Springer, 2006.

\bibitem[ER08]{EhrhardR08}
Thomas Ehrhard and Laurent Regnier.
\newblock Uniformity and the {T}aylor expansion of ordinary lambda-terms.
\newblock {\em Theor. Comput. Sci.}, 403(2-3):347--372, 2008.

\bibitem[GPR17]{GuerrieriPR17}
Giulio Guerrieri, Luca Paolini, and Simona Ronchi~Della Rocca.
\newblock Standardization and conservativity of a refined call-by-value
  lambda-calculus.
\newblock {\em Logical Methods in Computer Science}, 13(4), 2017.

\bibitem[Las99]{Lassen99}
S{\o}ren~B. Lassen.
\newblock Bisimulation in untyped lambda calculus: B{\"{o}}hm trees and
  bisimulation up to context.
\newblock {\em Electr. Notes Theor. Comput. Sci.}, 20:346--374, 1999.

\bibitem[Las05]{Lassen05}
S{\o}ren~B. Lassen.
\newblock Eager normal form bisimulation.
\newblock In {\em 20th {IEEE} Symposium on Logic in Computer Science {(LICS}
  2005), 26-29 June 2005, Chicago, IL, USA, Proceedings}, pages 345--354.
  {IEEE} Computer Society, 2005.

\bibitem[Mog92]{Mogensen92}
Torben~{\AE}. Mogensen.
\newblock Efficient self-interpretations in lambda calculus.
\newblock {\em J. Funct. Program.}, 2(3):345--363, 1992.

\bibitem[MP11]{ManzonettoP11}
Giulio Manzonetto and Michele Pagani.
\newblock B{\"{o}}hm's theorem for resource lambda calculus through {T}aylor
  expansion.
\newblock In C.{-}H.~Luke Ong, editor, {\em Typed Lambda Calculi and
  Applications - 10th International Conference, {TLCA} 2011, Novi Sad, Serbia,
  June 1-3, 2011. Proceedings}, volume 6690 of {\em Lecture Notes in Computer
  Science}, pages 153--168. Springer, 2011.

\bibitem[MRP19]{ManzonettoPR18}
Giulio Manzonetto, Simona Ronchi~Della Rocca, and Michele Pagani.
\newblock New semantical insights into call-by-value $\lambda$-calculus.
\newblock {\em Fundam. Inform.}, 170(1-3):241--265, 2019.

\bibitem[Ong97]{LukeOng}
Luke Ong.
\newblock Lambda calculus, 1997.
\newblock Lecture Notes.

\bibitem[Pao01]{Paolini01}
Luca Paolini.
\newblock Call-by-value separability and computability.
\newblock In Antonio Restivo, Simona Ronchi~Della Rocca, and Luca Roversi,
  editors, {\em Theoretical Computer Science, 7th Italian Conference, {ICTCS}
  2001}, volume 2202 of {\em Lecture Notes in Computer Science}, pages 74--89.
  Springer, 2001.

\bibitem[Pao08]{Paolini08}
Luca Paolini.
\newblock Parametric $\lambda$-theories.
\newblock {\em Theoretical Computer Science}, 398(1):51 -- 62, 2008.

\bibitem[Plo75]{Plotkin75}
Gordon~D. Plotkin.
\newblock Call-by-name, call-by-value and the lambda-calculus.
\newblock {\em Theor. Comput. Sci.}, 1(2):125--159, 1975.

\bibitem[PRDR99]{PaoliniR99}
Luca Paolini and Simona Ronchi Della~Rocca.
\newblock Call-by-value solvability.
\newblock {\em {ITA}}, 33(6):507--534, 1999.

\bibitem[RDRP04]{RonchiP04}
Simona Ronchi Della~Rocca and Luca Paolini.
\newblock {\em The Parametric $\lambda$-Calculus: a Metamodel for Computation}.
\newblock EATCS Series. Springer, Berlin, 2004.

\bibitem[Reg94]{Regnier94}
Laurent Regnier.
\newblock Une \'equivalence sur les lambda-termes.
\newblock {\em Theoretical Computer Science}, 126:281--292, 1994.

\bibitem[Roc18]{RonchiPC}
Simona Ronchi~Della Rocca.
\newblock Personal communication, 2018.

\bibitem[Tra09]{TranquilliTh}
P.~Tranquilli.
\newblock {\em Nets Between Determinism and Nondeterminism}.
\newblock PhD thesis, Univ. {Paris 7} and Univ. {Roma 3}, 2009.

\end{thebibliography}
\end{document}